\documentclass[a4paper,UKenglish,cleveref, autoref, thm-restate, numberwithinsect, anonymous,nolineno]{socg-lipics-v2021}
\usepackage{amsthm,amsmath,amssymb}
\usepackage[dvipsnames]{xcolor}
\usepackage{tikz}
\usepackage{graphicx} % Required for inserting images
\usepackage[font=small,labelfont=bf]{caption} % For captions
\usepackage[normalem]{ulem}
\usepackage{enumerate}
\usepackage{soul} % For underlines that wrap 
\usepackage{algorithm}
\usepackage{algpseudocode}

\hideLIPIcs
\newtheorem{question}{Question}

\newcommand{\Hyp}{\mathbb{H}}
\newcommand{\Reals}{\mathbb{R}}
\newcommand{\poly}{\mathrm{poly}}
\newcommand{\Oh}{\mathcal{O}}
\newcommand{\D}{\mathrm{d}}
\newcommand{\conv}{\mathrm{conv}}
\newcommand{\convGK}{G_K}

\newcommand{\dist}{\mathrm{dist}}
\newcommand{\area}{\mathrm{Area}}

\newcommand{\len}{\mathrm{len}}

\renewcommand{\angle}{\sphericalangle}

\newcommand{\interior}{\mathrm{Int}}

\title{Shortest Paths, Convexity, and Treewidth in Regular Hyperbolic Tilings}
\author{S\'andor Kisfaludi-Bak}{Department of Computer Science, Aalto University, Finland}{sandor.kisfaludi-bak@aalto.fi}{https://orcid.org/0000-0002-6856-2902}{Supported by the
Research Council of Finland, Grant 363444.}

\author{Tze-Yang {Poon}}{Department of Computer Science, University of Oxford, United Kingdom}{poontzeyang@gmail.com}{https://orcid.org/0009-0006-2260-9580}{}

\author{Geert {van Wordragen}}{Department of Computer Science, Aalto University, Finland}{geert.vanwordragen@aalto.fi}{https://orcid.org/0000-0002-2650-638X}{}

\authorrunning{S\'andor Kisfaludi-Bak, Tze-Yang {Poon}, Geert {van Wordragen}}
\Copyright{Anonymous author(s)} %TODO mandatory, please use full first names. LIPIcs license is "CC-BY";  http://creativecommons.org/licenses/by/3.0/
%TODO mandatory: Please choose ACM 2012 classifications from https://dl.acm.org/ccs/ccs_flat.cfm 
 %TODO mandatory; please add comma-separated list of keywords
\category{} %optional, e.g. invited paper
\relatedversion{} %optional, e.g. full version hosted on arXiv, HAL, or other respository/website
\acknowledgements{}%I want to thank \dots}%optional

\EventEditors{John Q. Open and Joan R. Access}
\EventNoEds{2}
\EventLongTitle{42nd Conference on Very Important Topics (CVIT 2016)}
\EventShortTitle{CVIT 2016}
\EventAcronym{CVIT}
\EventYear{2016}
\EventDate{December 24--27, 2016}
\EventLocation{Little Whinging, United Kingdom}
\EventLogo{}
\SeriesVolume{42}
\ArticleNo{23}
\begin{document}

\maketitle

\begin{abstract}
Hyperbolic tilings are natural infinite planar graphs where each vertex has degree $q$ and each face has $p$ edges for some $\frac1p+\frac1q<\frac12$. We study the structure of shortest paths in such graphs. We show that given a set of $n$ terminals, we can compute a so-called isometric closure (closely related to the geodesic convex hull) of the terminals in near-linear time, using a classic geometric convex hull algorithm as a black box. We show that the size of the convex hull is $\Oh(N)$ where $N$ is the total length of the paths to the terminals from a fixed origin.

Furthermore, we prove that the geodesic convex hull of a set of $n$ terminals has treewidth only $\max(12,\Oh(\log\frac{n}{p + q}))$, a bound independent of the distance of the points involved. As a consequence, we obtain algorithms for subset TSP and Steiner tree with running time $\Oh(N \log N) + \poly(\frac{n}{p + q}) \cdot N$.
\end{abstract}

\section{Introduction}

The uniform tilings of the three geometries (Euclidean, elliptical/spherical, and hyperbolic) are fundamental discrete structures that have received a lot of well-deserved attention in algorithms and discrete mathematics. In the geometric algorithms literature, these objects are studied in their own right, and they are also the basis of countless algorithmic techniques: e.g., a square grid is very often the basis of geometric approximation algorithms, but the triangular and hexagonal grids are also often sought out for their properties. In the spherical setting, the tilings correspond to the platonic solids; these are finite structures that have been used and studied for millennia.

\begin{figure}
    \centering
    \includegraphics{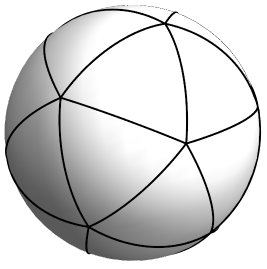}
    \includegraphics{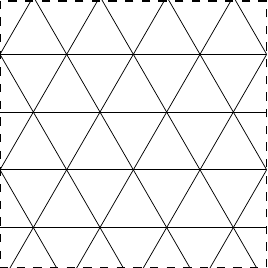}
    \includegraphics{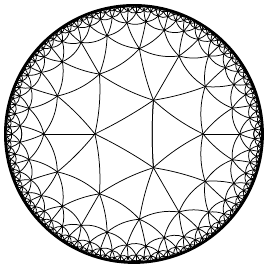}
    \caption{The regular tiling in spherical, Euclidean and hyperbolic geometry where $p=3$ and $q=5,6$, and $7$, respectively.}
    \label{fig:tilings}
\end{figure}

In the hyperbolic plane, unlike the other geometries, there is a much richer set of uniform tilings: for any pair of integers $p,q$ where $\frac{1}{p}+\frac{1}{q}<\frac12$ there exists a decomposition of the hyperbolic plane (henceforth denoted by $\Hyp^2$) into congruent copies of interior-disjoint regular $p$-gons (tiles) where each vertex is shared by $q$ tiles. This decomposition is called the uniform hyperbolic tiling of Schl\"afli symbol $\{p,q\}$.
An emerging body of work studies uniform hyperbolic tilings from the algorithmic perspective~\cite{Kisfaludi-Bak20-Intersection,dynamicdistances,Kopczynski21,Kisfaludi-BakML24}, but many fundamental questions, such as computing shortest paths, remain mostly unexplored.

In Euclidean square grids it is easy to characterize shortest paths: given a pair of points, these are the set of $x$- and $y$-monotone paths connecting the pair. These shortest paths are covered by the minimum bounding box of the point pair. The problem is much less clear when one wants to find a shortest path between two given vertices of a hyperbolic tiling. Due to the fact that $\Hyp^2$ is not a vector space (a typical pair of ``translations'' do not commute), it is significantly harder to find shortest paths in hyperbolic tilings. It is not even clear how one should define a vertex of a tiling in the input of an algorithm.

One way to represent a vertex is geometrically, by its coordinates in some model of the hyperbolic plane. Another is graph/group theoretically, as a path from a fixed vertex in the tiling graph or similarly as a word in the corresponding group (i.e.\ the group generated by the translations that take a vertex to its neighbors). Since these groups are non-abelian, conversion between these two representations is not as straightforward as for regular Euclidean tilings. However, the groups are still \emph{automatic}, which lets us convert a word of length~$\ell$ to a normal form in $\Oh(\ell^2)$ time \cite{wordprocessing,bridson2013metric} and thereby solve the \emph{word problem} (``do these words represent the same group element''), which is undecidable in general. In fact, they are \emph{strongly geodesically automatic}: words in the normal form correspond to shortest paths \cite{Cannon1984TheCS}. In this paper we will represent the input vertices of the tiling via paths (or even walks) from a fixed origin vertex; in the hyperbolic setting, this has a similar complexity as the number of bits in representations with coordinates, see \Cref{sec:preliminaries} for a discussion.
    
\subparagraph*{Shortest paths and intervals in graphs and tilings}
In the graph setting, \emph{geodesic convexity} has been studied intensively, see the monograph~\cite{pelayo2013geodesic} for a detailed overview of the topic. In graphs, there can be several shortest paths between a pair of endpoints, so in order to generalize the notion of convexity, we need some further terminology, introduced here only for unweighted graphs. The \emph{interval} of a vertex pair $u,v\in V(G)$ is the subgraph $I_G(u,v)\subset E(G)$ given by the union of \textit{all} shortest paths between $u$ and $v$.

An important property of shortest paths in hyperbolic tilings is that they stay together in the following sense: for any pair of points $u,v$ in the tiling, and any pair of shortest $(u,v)$-paths $P$ and $P'$ we have that each vertex of $P'$ is within distance $\Oh_{p,q}(1)$ from some vertex of $P$. This property holds much more generally, even for certain approximate shortest paths in the more general setting of Gromov-hyperbolic metric spaces. See~\cite{bridson2013metric} for a detailed exposition. As a result, we can already derive that an interval $I(u,v)$ in a hyperbolic tiling has constant-size vertex separators, unlike the Euclidean grid where the interval of a pair of points can be a square-shaped patch of the grid of arbitrary side length. In this article, we will get a more accurate description of intervals in hyperbolic tilings.

\subparagraph*{Convexity in graphs}
In the Euclidean setting, it is natural to talk about the bounding box of a set $P$ of grid points, and observe that all pairwise shortest paths in $P$ are contained in the bounding box.  This notion is similar to the geometric notion of \emph{convex hulls}.

Let $G$ be a graph. We say that a subgraph $H\subseteq E(G)$ is \emph{convex} if for any $u,v\in V(H)$ we have $I_G(u,v)\subseteq H$. On the other hand, there is a natural weaker property: a subgraph $H\subseteq E(G)$ is \emph{isometric} if for any $u,v\in V(H)$ we have $\dist_H(u,v)=\dist_G(u,v)$, where $\dist_X$ denotes the shortest-path distance in the (sub)graph $X$. Notice that any convex subgraph is automatically isometric, but this is not true the other way around: an isometric subgraph is guaranteed to contain \emph{some} shortest path between any pair of vertices, but not all of them.

Using the notions of convex and isometric subgraphs, we can define corresponding closures. For a given vertex set $K\subseteq V(G)$ a subgraph $H$ of $G$ is a \emph{convex hull} (respectively, a \emph{minimal isometric closure}) of $K$ if $H$ is a minimal convex (resp., minimal isometric) subgraph with $K\subseteq V(H)$. Notice that the convex hull of $K$ is in fact unique, while there can be many pairwise incomparable minimal isometric closures. See Figure~\ref{fig:gridhull} illustrating the difference between a minimal isometric closure and a convex hull in a Euclidean grid.

\begin{figure}
    \centering
    \includegraphics[scale=.75]{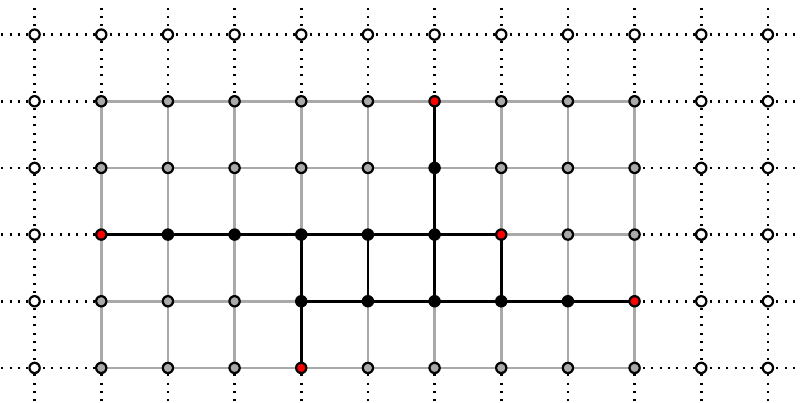}
    \caption{The convex hull (gray) and a minimal isometric closure (black) of a set of terminals (red) in the grid graph.}
    \label{fig:gridhull}
\end{figure}

We will be interested in the following question.

\begin{question}\label{q:convexhull}
    Is there an efficient algorithm to compute minimal isometric closures and convex hulls in hyperbolic tilings?
\end{question}

\subparagraph*{Optimization problems and convex hulls}

Beyond the basic distance properties and convexity, the subgraphs of hyperbolic tilings hold a significant algorithmic promise as a graph class.
Kisfaludi-Bak~\cite{Kisfaludi-Bak20-Intersection} observes that any $n$-vertex subgraph of a regular hyperbolic tiling $G_{p,q}$ has treewidth $\Oh_{p,q}(\log n)$ in stark contrast with the Euclidean setting where a $\sqrt{n}\times \sqrt{n}$ grid graph has treewidth $\Theta(\sqrt{n})$. This structural result combined with the literature on treewidth-based algorithms (see the book by Cygan~et~al.~\cite{CyganFKLMPPS15} for an overview) gives polynomial algorithms in this graph class for several problems that are NP-hard on planar graphs. Still, the result is not completely satisfactory: first, it is not clear how one can recognize that a graph is indeed a subgraph of a tiling (this remains a challenging open question), and second, in many cases planar problems are more naturally defined via the point set rather than with a large finite grid graph.
For example, the rectilinear TSP and Steiner tree problems have as their input a set of $n$ points, and the goal is to compute the shortest closed curve or tree, consisting only of horizontal and vertical segments, that contains all the input points. In the grid setting, these problems can be thought of as being defined on a (weighted) $n\times n$ grid given by the horizontal and vertical lines through the input points, often called the Hanan-grid~\cite{HananGraph}. The best known algorithms for these problems have a running time of $k^{\Oh(\sqrt{k})}n^{\Oh(1)}$ or $n^{\Oh(\sqrt{k})}$ in the Euclidean plane~\cite{FominSubexpRectSteiner,KleinM14}, where $k$ is the number of terminals.

In the graph setting, the Subset TSP problem asks for the shortest closed walk containing a given set of terminals, and Steiner tree asks for the shortest tree containing the terminals. In our setting of an infinite host tiling, we can restrict our attention to the convex hull of the terminals. Intuitively, the number of tiling vertices inside the convex hull is  similar to the geometric area of the hull; indeed, we can show that the minimal isometric closure has area that is \emph{linear} as a function of the size of the minimum spanning tree of the terminal set, which comes from the linear isoperimetric inequality~\cite{bridson2013metric} that is unique to the hyperbolic setting. Using the logarithmic treewidth bound on the minimal isometric closure, one should be able to get polynomial algorithms for subset TSP and Steiner tree via a black-box usage of treewidth-based algorithms for these problems~\cite{BodlaenderCKN15}. It is however unclear if it is possible to get an algorithm that is near-linear as a function of the diameter, and polynomial in the number of terminals.

\begin{question}\label{q:algo}
    Is there an algorithm to compute Subset TSP and Steiner tree in regular hyperbolic tilings that is (near)-linear in the input bit complexity and polynomial in the number of terminals?
\end{question}

\subparagraph*{Our contribution}
We start by studying the structure of shortest paths in the hyperbolic tiling graph $G_{p,q}$. For a line $\ell$ in the hyperbolic plane we show a lemma that can be informally stated as follows.

\begin{lemma}[Informal, weaker version of \Cref{lemma:shortestPath}(i)]
   For any pair $u,v$ of vertices incident to tiles intersected by $\ell$ there exists a shortest path from $u$ to $v$ whose edges are all incident to tiles that intersect $\ell$.
\end{lemma}

The proof of the lemma is based on analyzing the hyperbolic area enclosed by a hypothetical shortest path that encloses several tiles between $\ell$ and itself. We can use a similar type of area argument to show that the interval $I(u,v)$ is covered by a sequence of tiles, i.e., each vertex of an interval is on the boundary of the subgraph induced by the interval.

The above lemma is insufficient for our purposes, as we need to be able to extend a shortest path of the tiling ``along'' a hyperbolic line by adding new edges at one end so that the result is still a shortest path. When $q=3$, it is common that the line $\ell$ can intersect all three tiles incident to some vertex; for a fixed shortest path it is unclear if any of its extensions remain a shortest path in this situation. We prove the following stronger lemma.

\begin{lemma}[Informal, weaker version of \Cref{lemma:extensionByIntersection}] \label{lemma:weakExtensionByIntersection}
   Let $S$ be the sequence of edges of the tiling $G_{p,q}$ intersected by some line $\ell$. Then for any pair of vertices $v,w$ that are endpoints of such edges there exists a shortest path from $v$ to $w$ that passes through at least one endpoint of each edge of $S$ between the edge of $v$ and the edge of $w$.
\end{lemma}

The proof of \Cref{lemma:extensionByIntersection} is significantly more technical than that of \Cref{lemma:shortestPath}(i) and relies heavily on the specific geometry of the tilings and the ways in which a line can intersect consecutive tiles. With these lemmas and some bound on the intervals $I(u,v)$ at hand, we are able to compute the shortest path from $u$ to $v$ (or even all shortest paths) in time $\Oh(\dist_{G_{p,q}}(u,v))$ time.

With the stronger lemma at hand, we can move on to computing isometric closures, but in order to do that, we need to fix how we receive the input points. We denote by $K$ the set of terminals from a tiling graph $G_{p,q}$, and let $n=|K|$. We assume that the terminals are defined in the graph or group-theoretic sense, i.e., each terminal $v\in K$ is defined via a path (or walk) in $G_{p,q}$ that starts at the fixed origin and ends at the terminal. We note that for constant $p,q$, the number $N$ is within a constant factor of the bit complexity if the points were given by coordinates in the half-space model, and we can make similar claims about other hyperbolic models; for a more in-depth discussion of different input modalities and our computational model see \Cref{sec:preliminaries}.

Towards \Cref{q:convexhull}, we show that we can compute an isometric closure $\hat G_K$ of $K$ that is a subgraph of the convex hull $\conv_G(K)$ in $\Oh(N\log N)$ time. Note that \cite{pelayo2013geodesic} gives an algorithm that computes the graph convex hull of $n$ terminals in any graph with $m$ edges in $\Oh(mn)$ time, which is already optimal (under SETH) for recognizing if a given subgraph is convex \cite{Cabello25}. In our setting we have an infinite base graph, and the simplest restriction to a finite graph (taking all vertices within a given distance of one fixed vertex) can have $q^{\Theta(N)}$ edges.
We show that the convex hull $\conv_G(K)$ has $\Oh(N)$ vertices. Furthermore, we bound the treewidth of $\conv_G(K)$ as follows.

\begin{restatable}{theorem}{thmtw}
\label{thm:treewidth}
    For any set $K$ of $n$ vertices in $G_{p,q}$, the convex hull $\conv_G(K)$ has treewidth at most $\max\{12, \Oh(\log\frac{n}{p+q})\}$.
\end{restatable}

This can be seen as a far-reaching strengthening of the treewidth bound of~\cite{Kisfaludi-Bak20-Intersection}, which was reminiscent of bounds for other graphs that come from hyperbolic geometry \cite{ChepoiDEHV08, BlasiusFK16, structureindependence}.
Indeed, the treewidth depends logarithmically on the number of terminals and it is independent of the size of the convex hull, and we get stronger bounds for larger values of $p$ and $q$. This is in line with the observation that hyperbolic structure becomes more tree-like and hence simpler at larger distances, see~\cite{hyperTSP20,Kisfaludi-BakW24,structureindependence} for further examples of this phenomenon.

The above treeewidth bound automatically holds also for the isometric closure $\hat G_K\subseteq \conv_G(K)$ that we compute. We prove that any minimal isometric closure of $K$ contains a solution to Subset TSP and Steiner Tree, so we are able to compute an exact solution to these problems in $\hat G_K$. Using existing treewidth-based algorithms, we get the following answer to \Cref{q:algo}.

\begin{restatable}{theorem}{algomain}\label{thm:algomain}
    Given a set of $n$ terminals with total description size $N$ in a regular hyperbolic tiling graph $G_{p,q}$ with Schläfli symbol $\{p,q\}$, the Steiner tree and Subset TSP problems can be solved in $\Oh(N \log N) + \poly(\frac{n}{p + q}) \cdot N$ time.
\end{restatable}

This is a significant strengthening compared to the general planar graph setting, where the problem is NP-hard and the best known algorithms are subexponential~\cite{KleinM14,DBLP:conf/focs/MarxPP18}. Note that due to the logarithmic treewidth bound~\cite{Kisfaludi-Bak20-Intersection} one naturally expects a running time that is polynomial in the size of the underlying subgraph (here, the convex hull); the main contribution of the above theorem is that it is \emph{near-linear} in the size $N$ of the convex hull and depends polynomially only on the number $n$ of terminals. Moreover, the algorithm becomes faster as $p$ or $q$ grows, reaching an $\Oh(N\log N)$ algorithm when $\max(p,q)=\Omega(n)$.

% (Sketch of overall approach): First, show that there always exists an optimal Steiner tree within any planar geodesic subgraph that contains all terminals. (Insert statement about $2^{\text{\Oh(max distance from origin)}}$ solution here). However, we will improve this technique by generating a geodesic subgraph with a low treewidth. The subgraph we generate will by ``close" to the convex hull of the terminals in the hyperbolic plane. This subgraph can be generated in $\mathcal{O}(n\log(n) + N)$ time, has treewidth $\log(n)$ and $\mathcal{O}(N)$ vertices. This will enable us to use treewidth-based parameterised algorithms for Steiner tree to solve the problem in the geodesic subgraph in time $\poly(n)\cdot N$ time. (Further reduction to $\poly(n)+ \mathcal{O}(N)$ time).

% From the work Garey and Johnson \cite{rectilinearSteinerTree} and Hanan \cite{HananGraph}, it is known that Steiner tree is \textsc{NP}-complete on the graph of the square tiling of the Euclidean plane. Hence, it seems to be necessary to exploit the ``tree-like" nature of the hyperbolic plane to achieve our results.

\section{Preliminaries}\label{sec:preliminaries}

\subparagraph*{Hyperbolic Geometry}
Apart from basic graph theory and geometry, this article uses the few key properties of hyperbolic geometry given below. For a more in-depth understanding of hyperbolic geometry, see \cite{cannonhyperbolic} or the larger textbooks \cite{iversen1992hyperbolic,thurston97three,benedetti1992lectures}.

For any two points $u,v$ in the hyperbolic plane, there is a unique line segment $uv$ joining them and a unique line $\ell_{uv}$ that is incident to both points. For any three points $u,v,w$, the area of the triangle $\triangle{uvw}$ with internal angles summing to $\phi$ is $\pi -\phi$~\cite{trig_formulas}. A polygon with $m$ vertices and internal angles summing to $\phi$ has an area of $ \pi(m-2) - \phi$. Hence, the area of a tile in $G_{p,q}$ is $\pi p(1-2/q) - 2\pi$. Note that for $p,q\geq3$ this area is lower bounded by some positive constant. 

To visualize the hyperbolic plane, one can use a \emph{model}.
One such model is the Beltrami-Klein model (also known as the Klein disk), which assigns each point of $\Hyp^2$ coordinates inside the unit disk in $\Reals^2$. Distances and angles are heavily distorted, especially near the disk boundary, but one key property of the model is that for any points $u,v \in \Hyp^2$ the hyperbolic line segment $uv$ is also a Euclidean line segment. This for example means that the hyperbolic convex hull $\conv_\Hyp$ matches the Euclidean convex hull in the Beltrami-Klein model.
In our figures we use (or mimic) the Poincaré disk model, which similarly assigns coordinates inside a unit disk but uses a different distance function. As a consequence, it now keeps Euclidean and hyperbolic angles equal, but hyperbolic segments $uv$ are given by specific Euclidean circular arcs.
Note, however, that all our results are independent of the particular model of the hyperbolic plane and conversion between models is easy when one has access to real RAM with standard arithmetic and square roots \cite{cannonhyperbolic}.

\subparagraph*{Regular Tilings of \texorpdfstring{$\Hyp^2$}{ℍ²}}
A regular tiling of the hyperbolic plane is an edge-to-edge filling of the hyperbolic plane with regular polygons as its faces. Each tiling can be identified with a Schläfli symbol, where a $\{p,q\}$-tiling refers to a tiling comprising regular $p$-gons where $q$ faces meet at every vertex. There exists a hyperbolic tiling for each $\{p,q\}$ where $1/p+1/q<1/2$. Let $G_{p,q}$ refer to the graph derived from the $\{p,q\}$-tiling (i.e. using the same vertices and edges). When $p$ and $q$ are not relevant we use $G$ for succinctness. Note that $G$ is planar and we identify a vertex $v\in G$ with the point of $\Hyp^2$ in a fixed representation.

In this paper, we refer to a regular polygon in the $\{p,q\}$-tiling as a \textit{tile}. The edges and vertices of each finite subgraph of $G_{p,q}$ splits $\Hyp$ into \textit{faces}, and the face with unbounded area is referred to as the \textit{unbounded face}. Given a closed curve $\gamma$, denote the union of the bounded faces by $F_\gamma$. Denote the interior of a collection of bounded faces $B$ as $\interior(B)$. %An edge/vertex is \textit{incident} to a tile if it is a subset of the boundary of the tile. A vertex is incident to an edge if it is one of the two endpoints of that edge. The relative interior of an edge excludes the two endpoints of that edge.

\subparagraph*{Treewidth and Outerplanarity}
A \textit{tree decomposition} \cite{treewidth} of a graph $G$ is a tree $T$ in which each node $x$ has an assigned set of vertices $B_x\subseteq V$ such that $\bigcup_{x\in T}B_x=V$ where:
\begin{itemize}
    \item for any $uv\in E$, there exists a $B_x$ such that $u,v\in B_x$,
    \item if $u\in B_x$ and $u\in B_y$ then $u\in B_z$ for any $z$ on the (unique) $(x,y)$-path in $T$.
\end{itemize}
The \emph{treewidth} of $G$ is now the smallest value of $\max_{x \in V} |B_x| - 1$ over all tree decompositions.

An embedding in the plane of a graph $G$ is \textit{outerplanar} (or 1-\textit{outerplanar}) if it is planar and all vertices lie on the unbounded face. An embedding of $G$ is $k$-\textit{outerplanar} if it is planar and deleting all vertices from the unbounded face leaves a $(k-1)$-outerplanar embedding of the remaining graph. A graph is $k$-outerplanar if it admits a $k$-outerplanar embedding.

\subparagraph*{Input Representation and Computational Model}

% In the Steiner tree problem on graphs, we are given a graph and a subset of vertices referred to as terminals. The objective is to find the smallest tree in the graph that contains all terminals.
In the setting of hyperbolic tiling graphs, the entire graph can be specified simply by giving the associated Schläfli symbol. Each problem instance can then be specified by the number of terminals as well as the locations of these terminals. In this paper, we assume that each of the $|K| = n$ terminals is specified by a walk of edges from a fixed starting vertex, which we will assume to be at the origin of our Poincar\'e model. We fix a clockwise enumeration of the edges around the origin. At each step of the walk we consider the previous edge as being the first in the clockwise enumeration. Then each step of the walk can continue along one of $q$ possible incident edges, and a path of length $t$ can be encoded as sequence of $t$ numbers from $\{1,\dots,q\}$. We use $N$ to refer to the total length of the walks used to describe the location of all terminals.

As mentioned in the introduction, $N$ is similar to the bit complexity if the terminals $K$ were given by their coordinates in some model of the hyperbolic plane. This can be seen as follows: the number of vertices within $r$ hops from a given vertex is exponential in $r$, so representing all these vertices with unique coordinates requires $\Omega(r)$ bits. Thus, for constant $p,q$ and in the average case, the length of the shortest path to a vertex and the number of bits required for the vertex's coordinates are within a constant factor of each other.

To avoid dealing with issues of precision, we assume we have access to a real RAM.
In addition to standard arithmetic operations, we need this machine to support the square root and sine function to generate tiling coordinates.
Formulas for generating tiling coordinates are given for example in~\cite{DUNHAM1986139} and can also be deduced from hyperbolic trigonometry in the Poincaré disk model.

\section{Shortest Paths in the Tiling Graph} \label{sec:shortestpaths}

% \typ{This intro is no longer accurate.}

% In this section, we describe how we will restrict the tiling graph $G$ to a subgraph $\convGK$ that contains an optimal Steiner tree.

% Our approach can be understood by analogy with Steiner tree on the Euclidean plane. In the Euclidean setting, there is always an optimal Steiner tree contained within the convex hull of the terminals. Intuitively, this is because any Steiner tree that has a subtree outside the convex hull can have that subtree replaced by the curve traversing the boundary of the convex hull between the two furthest-apart leaves of that subtree. Similarly, $\partial\convGK$ will be defined using a closed sequence of shortest paths between terminals and $\convGK$ will be the subgraph of vertices in the union of edges and unbounded faces.

Let us start our investigation of shortest paths in tilings for point pairs that are vertices of the same tile.

\begin{lemma} \label{lemma:singleFaceShortestPath}
    For any two vertices $u$ and $v$ incident to a tile $t$ in $G$, any shortest $(u,v)$-path must only use edges incident to $t$.
\end{lemma}

\begin{proof}
    Consider the rays beginning in the center of $t$ and each incident to a different vertex of $t$. Truncate each ray so that it starts from the corresponding vertex on $t$. For any two rays incident to adjacent vertices in $t$, the unique shortest curve that connects one ray to the other is exactly the edge of $t$ connecting these two rays. Hence, any $(u,v)$ path that uses any edges that are not incident to $t$ must be strictly longer than the shortest path on the boundary of~$t$.
\end{proof}
\begin{figure}[t]
    \centering
    \includegraphics[width=0.7\linewidth]{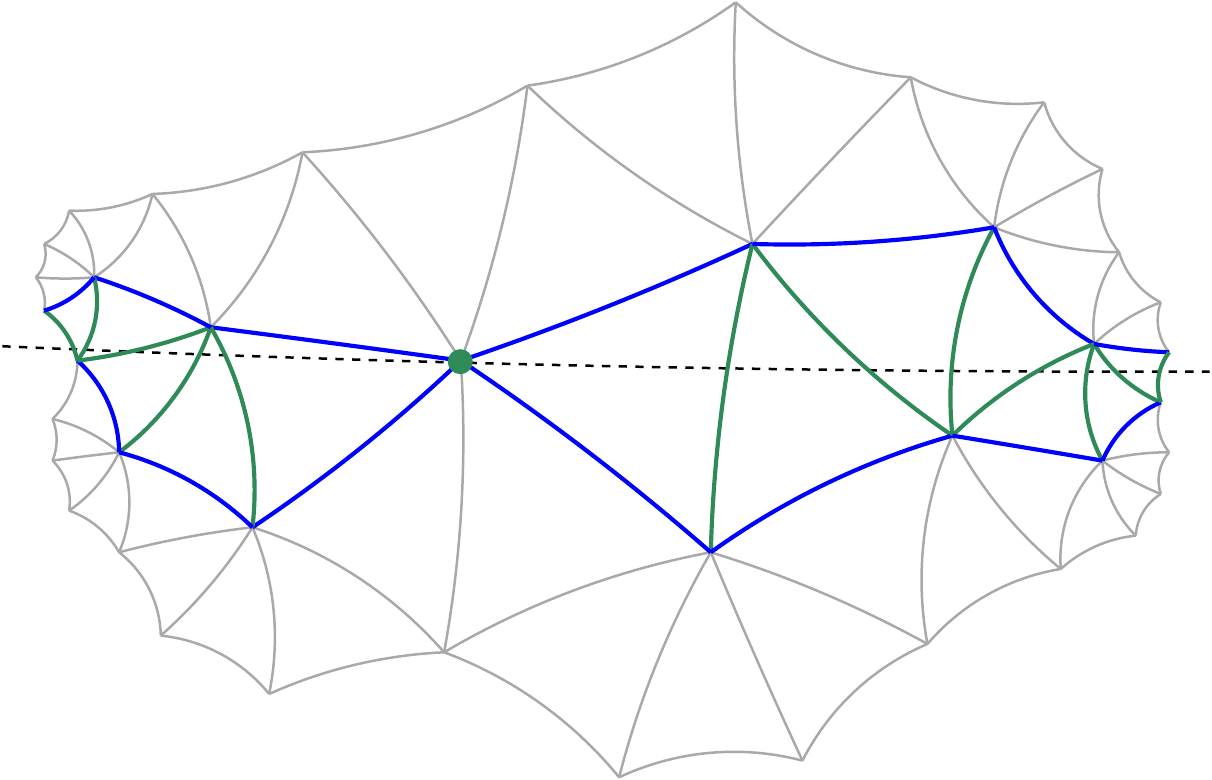}
    \caption{The subgraph $G_\ell$ (blue and green) intersected by a line $\ell$ (dashed) and the sequence~$S_\ell$ of vertices and edges (green) intersected by $\ell$. An additional layer of nearby tiles is depicted (grey).}
    \label{fig:Gl}
\end{figure}
\begin{definition} [Subgraph intersected by a line]
    Let $\ell$ be a line in $\Hyp^2$. The subgraph intersected by $\ell$, $G_{\ell}$, is the subgraph of $G$ induced by the edges that are in $\ell$ and edges incident to tiles whose interiors are intersected by $\ell$. 
\end{definition}

\begin{figure}
    \centering
    \includegraphics{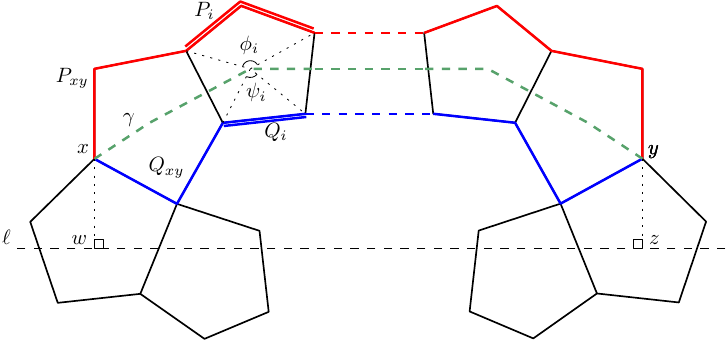}
    \caption{The path $P_{xy}$ (red) cannot be shorter than $Q_{x,y}$ (blue). Only the solid lines represent edges of $G$.}
    \label{fig:Gl}
\end{figure}

\newcommand{\lemmaShortestPathPartTwo}{For any two vertices $x,y$ and any two shortest $(x,y)$-paths $P_1,P_2$. Then there does not exist a vertex $w$ such that $w\in \interior(F_{P_1\cup P_2})$.

Furthermore, for each vertex $v\in V(P_1\cup P_2)$, the number of incident tiles in $F_{P_1\cup P_2}$ is no greater than $4$ when $p=3$, $3$ when $p=4$ and $2$ when $p\geq 5$.}

\begin{lemma} \label{lemma:shortestPath}
 
\begin{enumerate} [(i)]
    \item For any line $\ell$ and any $u,v \in V(G_\ell)$, there exists a shortest path from $u$ to $v$ in $G$ that is fully contained in $G_\ell$.
    \item \lemmaShortestPathPartTwo
\end{enumerate}
    
\end{lemma}
\begin{proof}
\textbf{(i)} Suppose for contradiction that there is a shortest path $P_{u,v}\not\subseteq E(G_\ell)$ from $u$ to $v$ that is shorter than any path that stays in $G_{\ell}$, as depicted in Figure \ref{fig:Gl}. Define the \emph{boundary} of $G_\ell$, $\partial G_\ell$, to be the set of edges incident to the unbounded faces of $G_\ell$ that does not contain $\ell$. Let $x$ and $y$ be the vertices where $P_{u,v}$ first leaves and rejoins $G_{\ell}$ respectively, and let $P_{x,y}$ be the subpath of $P_{u,v}$ from $x$ to $y$. Without loss of generality, let $P_{x,y}$ proceed clockwise around the enclosed tiles from $x$ to $y$. Consider the subpath $R_{x,y}\subseteq\partial G_\ell$ between $x$ and $y$. Then $P_{x,y}\cup R_{x,y}$ forms a closed curve that \emph{encloses} a set of tiles (i.e. the tiles in the bounded region of each closed curve). Without loss of generality, select $P_{u,v}$ to be the path that also encloses the fewest such tiles across all maximal subpaths of $P_{u,v}$ which do not contain intermediate vertices in $G_{\ell}$.

    Let $x=v_0,v_1,\dots,v_k=y$ be the sequence of vertices along $P_{x,y}$.
    We construct a \textit{sequence of tiles following $P_{x,y}$}, starting from the enclosed tile incident to the edge $v_0v_1$ and ending in the enclosed tile incident to the last edge $v_{k-1}v_k$. Starting from $v_i = v_1$, add tiles incident to $v_i$ in the counterclockwise direction, incrementing $v_i = v_{i+1}$ when the tile incident to $v_{i+1}$ is added, until $v_i = v_k$. By Lemma \ref{lemma:singleFaceShortestPath}, no tile appears twice in this sequence, otherwise we could construct a shorter path than $P_{x,y}$ by shortcutting $P_{x,y}$ with edges around such a tile. Let $t_1, \dots , t_m$ be the resulting sequence of tiles. Let the $(x,y)$-path formed from the boundary of this set of tiles with $P_{x,y}$ removed be $Q_{x,y}$. Since $Q_{x,y}$ is an $(x,y)$-path which encloses fewer tiles than $P_{x,y}$ with $R_{x,y}$, the path $P_{x,y}$ must be strictly shorter than $Q_{x,y}$.

    Let the centres of the above sequence of tiles be $ c_1,\dots,c_m$. Consider the curve $\gamma$ comprising line segments joining the sequence of points $x, c_1,\dots,c_m ,y$. For each tile $t_i$, its edges can be partitioned into the edges which are in $P_{x,y}$, which we refer to as $P_i$, the (at most two) edges whose relative interiors intersect $\gamma$ and the remaining edges $Q_i$. The curve $\gamma$ makes two angles at each $c_i$, one on the side of $P_i$ and the other on the side of $Q_i$, which we refer to as $\phi_i$ and $\psi_i$, respectively, such that $\phi_i+\psi_i = 2\pi$. Note that $\len(P_i) - \len(Q_i)$ is proportional to $\phi_i - \psi_i$. Each edge in $P_{x,y}$ is incident to exactly one $t_i$ and the sum of the lengths of all remaining edges in $Q_1,\dots,Q_m$ is no smaller than the length of $Q_{x,y}$. Since the length of $P_{x,y}$ is strictly shorter than $Q_{x,y}$, therefore $\sum_{i=1}^m \psi_i > \sum_{i=1}^m \phi_i$. Furthermore, since $\sum_{i=1}^m \left(\phi_i+\psi_i\right) = 2\pi m$ and each edge in $Q_{i}$ adds $2\pi/p$ to $\psi_i$, we have that $\sum_{i=1}^m \psi_i \geq \pi m + 2\pi/p$.
    
    Let the perpendiculars from $x$ and $y$ to $\ell$ intersect $\ell$ at $w$ and $z$ respectively. Consider the polygon $A$, whose edges are $\gamma \cup \{xw,wz,zy\}$. Polygon $A$ is simple as $\gamma$ lies outside of $G_\ell$ while the remaining edges are inside $G_\ell$. The internal angles at $x$ and $y$ are at least $2\pi/2q$ respectively, since the first and last tiles in the sequence are not in $G_{\ell}$. We now apply the formula for the area of a hyperbolic polygon given its internal angles on $A$:
    \begin{equation*}
        \area(A)\leq(m+4-2)\pi - m\pi - \frac{2\pi}{p} - 2\left(\frac{\pi}{2}\right) - 2\left(\frac{2\pi}{2q}\right) = \left(1 - \frac{2}{p} - \frac{2}{q}\right) \pi.
    \end{equation*}

    However, each tile can be split into $2p$ right-angled triangles by adding one line segment from the centre to each vertex and the middle of each edge. Calculating the area of a single right-angled triangle:
    \begin{equation*}
    \begin{split}
        \area(\triangle)& = \pi - \frac{\pi}{2} - \frac{2\pi}{2p} - \frac{2\pi}{2q} = \left(\frac{1}{2}-\frac{1}{p}-\frac{1}{q}\right)\pi.
    \end{split}
    \end{equation*}

    Finally, the length of $Q_{x,y}$ is at least 2 edges (at least 1 more than $P_{x,y}$, which has at least 1 edge), which means $A$ contains at least 4 such right-angled triangles, but this contradicts the upper bound on the area of $A$ above. See \Cref{sec:appsec3} for the proof of (ii).
\end{proof}

\begin{corollary}
    The interval of any two vertices is outerplanar.
\end{corollary}

In order to efficiently compute a shortest path, we show that it suffices to consider paths that contain the vertices intersected by $\ell$ and at least one endpoint of edges whose relative interiors intersect $\ell$.

\begin{restatable}{lemma}{lemmaBottlenecks} 
\label{lemma:bottleneck}
Let $\ell$ be a directed line in $\Hyp^2$.
Then, we can make the following observations relating to tiles intersected by $\ell$:
\begin{enumerate}[(i)]
    \item
    Let $q \geq 4$ and $t_1,t_2$ be tiles intersected by $\ell$ that share a vertex $v$.
    Then, $\ell$ must intersect an edge of $t_1$ incident to $v$.
    Specifically, if $t_1$ and $t_2$ share an edge $e$ then $\ell$ will intersect~$e$.

    \item 
    Let $q=3$ and $v$ be a vertex whose incident tiles $t_1, t_2, t_3$ are intersected by $\ell$, in that order.
    Then, $t_1$ does not share a vertex with any tile that $\ell$ intersects after $t_3$.

    \item 
    Let $q=3$ and $e$ be an edge intersected by $\ell$ with incident tiles $t_1, t_2$, whose shared neighbors are not intersected by $\ell$.
    Then, $t_1$ does not share a vertex with any tile that $\ell$ intersects after $t_2$.
\end{enumerate}
\end{restatable}

\begin{definition} [Sequence of edges/vertices intersected by $\ell$, $S_{\ell}$]
    Let $\ell$ be a line in $\Hyp^2$. Define $S_\ell$ to be the infinite sequence of vertices that $\ell$ is incident to and edges that have their relative interiors intersected by $\ell$. 
    
    Let $v,w$ be vertices either on $\ell$ or incident to edges that are not contained in $\ell$ and have their relative interiors intersected by $\ell$. Let $s_v,s_w\in S_\ell$ be the elements of S that are either the vertices $v$ and $w$ or are edges incident to $v$ and $w$ respectively. Define $S_{vw}$ to be the contiguous subsequence of $S_\ell$ that begins with $s_v$ and ends with $s_w$.
\end{definition}

\begin{restatable}{lemma}{lemmapain} \label{lemma:extensionByIntersection}
    Let $v,w$ be vertices that belong to $s_v,s_w\in S_\ell$ respectively. There exists a shortest $(v,w)$-path that intersects all elements of $S_{vw}$ in the order that they appear in $S_{vw}$. In particular, if $uu'$ is an edge in $S_{v,w}$, then $u$ or $u'$ is on a shortest $(v,w)$-path.
\end{restatable}

\section{Isometric Subgraph Properties}\label{sec:isomsubgraph}
Using Lemma~\ref{lemma:extensionByIntersection} we can now give some further results about shortest paths, which we will use to prove properties for any isometric subgraph of a regular hyperbolic tiling.

\begin{lemma} \label{thm:geodesic extension}
    If $P$ is a shortest $(a,b)$-path in $G$, and $C$ is a simple cycle in $G$ where $b$ is in the bounded region of (or on) $C$, then there exists $c\in V(C)$ and a shortest path $Q$ from $a$ to $c$ that contains $P$.
\end{lemma}

\begin{proof}
    The ray $\overrightarrow{ab}$ must intersect $C$. If it intersects $C$ at a vertex, let that vertex be $c$. Otherwise $\overrightarrow{ab}$ intersects $C$ at an edge, in which case let either vertex incident to that edge be $c$. Note that $b\in S_{\ell_{ab}}$ and either $c\in S_{\ell_{ab}}$ or $c$ is incident to an edge $e\in S_{\ell_{ab}}$. Hence by Lemma \ref{lemma:extensionByIntersection}, there is a shortest $(a,c)$-path $Q$ which contains $P$.
\end{proof}

Note that since \Cref{lemma:weakExtensionByIntersection} does not guarantee that $Q$ intersects $c$, it is therefore insufficient for the above proof. 

\begin{corollary} [Geodesic extension] \label{cor:geodesicExtensionByOne}
    If $P$ is a shortest $(a,b)$-path in $G$, then $P$ can be extended by one edge to a vertex $c$ such that $P\cup\{bc\}$ is a shortest $(a,c)$-path.
\end{corollary}

\begin{proof}
    Let $C$ be the the cycle that encloses all tiles incident to $b$ and apply \Cref{thm:geodesic extension}.
\end{proof}

% \begin{definition} [Geodesic extension of a path]
%     A geodesic extension of a shortest $(a,b)$-path $P$ in $G$ is a shortest path $Q=P\cup\{bc\}$ for some edge $bc$.
% \end{definition}

We now work towards showing that there exists a Steiner tree and subset TSP walk that is optimal for $G$ contained in any fixed minimal isometric closure of $K$. For the rest of this section, let $G_K$ be an arbitrary minimal isometric closure of $K$.

\begin{restatable}[Isometric closures are hole-free]{lemma}{lemmaHoleFree}  \label{lem:isometricholefree}
    If $H$ is an isometric subgraph of $G$, then each bounded face of $H$ is a tile.
    
\end{restatable}

\begin{definition} [Boundary and boundary walk]
    Given a subgraph $H$ of $G$, the boundary of $H$, $\partial H$, is the set of edges that lie on the unbounded face of $H$.
    
    For $u \in V(\partial H)$, denote by $W_u$ the unique closed walk beginning at $u$ using all the edges in $\partial H$ that traverses around $H$ clockwise; we call this a boundary walk. The partial boundary walk $W_{uv}$ for a given $v \in V(\partial H)$ is the minimum prefix of $W_u$ ending at $v$.
\end{definition}

\begin{restatable}[Boundary of $\convGK$ between terminals] {lemma}{lemmaBoundaryBetweenTerminals}\label{lemma:boundaryBetweenTerminal}
    If $s,t\in K$ are terminals on $\partial\convGK$ such that there are no other terminals on the walk $W_{st}$, then $W_{st}$ is a shortest $(s,t)$-path.
\end{restatable}

\begin{restatable}{lemma}{lemmaBoundaryWalk}\label{lemma:boundaryWalkLemmas}
    Given vertices $u,v\in V(\partial\convGK)$, let $b$ be the next vertex after $v$ on $W_u$.
    \begin{enumerate}[(i)]
        \item If $W_{uv}$ is a shortest $(u,v)$-path, then either $b$ is a geodesic extension of $W_{uv}$ or for all vertices $s$ which are geodesic extensions of $W_{uv}$, $b$ is further clockwise about $v$ from $W_{uv}$  than $s$.
        \item If $Q_{uv}$ is a $(u,v)$-path which does not use any edges in $\convGK$, then $\len(Q_{uv})\geq \len(W_{uv})$.
    \end{enumerate}
\end{restatable}

\begin{figure}
    \centering
    \includegraphics{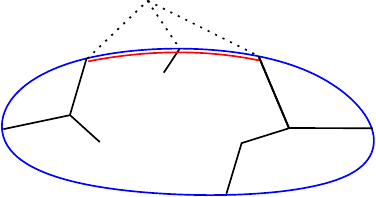}
    \caption{Optimal Steiner tree $T$ (black) can be made to use only edges in $\convGK$ (blue) by replacing subtrees $R_i$ (dotted black) with boundary walks (red).}
\end{figure}
\begin{lemma} \label{lemma:convGK}
    There exists a Steiner tree and subset TSP walk in $\convGK$ that is optimal for $G$.
\end{lemma}

\begin{proof}
    Let $T$ be any optimal Steiner tree in $G$. Removing any edges in $\convGK$ from $T$ produces a forest $R$. Let $R_i$ be a maximal subtree in $R$ such that none of its internal vertices are in $V(\partial\convGK)$. Let $V_i = V(R_i)\cap V(\partial\convGK)$. Note that $|V_i|\geq 2$ since $R_i$ intersects $\partial\convGK$ at $\geq2$ points. Let $u_i,v_i\in V_i$ be the two vertices furthest apart along $\partial\convGK$. Since $R_i$ contains a $(u_i,v_i)$-path, by \Cref{lemma:boundaryWalkLemmas} (ii) replacing $R_i$ with $W_{u_iv_i}$ produces a tree that is no larger. Repeating this for all $R_i$ removes all edges of $R$ outside $\convGK$. Furthermore, note that this process produces a Steiner tree since we only disconnect vertices outside of $\convGK$, which cannot be terminals. Therefore, this constructs an optimal Steiner tree using only edges from $\convGK$.

    The proof for Subset TSP is analogous: we set $R_i$ to be a maximal subwalk of some optimum tour that falls outside $\convGK$. Then by \Cref{lemma:boundaryWalkLemmas} (ii) it can be replaced with the corresponding part of the boundary walk of $\convGK$ without making the walk longer.
\end{proof}

\section{Computing an Isometric Closure in the Tiling Graph} \label{sec:computing}

Algorithm \ref{alg:boundary} is a procedure to explicitly compute the boundary of an isometric closure, which we denote $\hat{G}_K$. This is done by computing a shortest path between each pair of terminals $v_i$ and $v_{i+1}$ that defines a line segment $v_iv_{i+1}$ on the boundary of $\conv_\Hyp(K)$. In particular, the algorithm finds the shortest $(v_i,v_{i+1})$-path that intersects all elements of $S_{v_iv_{i+1}}$ in order. Due to Lemma \ref{lemma:singleFaceShortestPath}, the shortest paths between vertices belonging to consecutive elements of $S_{v_iv_{i+1}}$ consist only of edges of the shared tile. Hence, it suffices to compute the shortest path to each $s_j\in S_{v_iv_{i+1}}$ by dynamic programming. Furthermore, the vertices belonging to each $s_j\in S_{v_iv_{i+1}}$ can be found efficiently given the previous $s_{j-1}$ by traversing the shared tile in both directions and checking whether we have reached the next tile intersected by $v_iv_{i+1}$. In total over all of $S_{v_iv_{i+1}}$ this traverses a number of edges that is at most twice the length of a shortest $(v_i,v_{i+1})$-path.

\begin{restatable}{lemma}{lemmaHatGKProperties}
    The graph $\hat{G}_K$ is isometrically closed and is a subgraph of the convex hull $\conv_G(K)$.
\end{restatable}

\begin{corollary}
    There exists a minimal isometric closure $G_K$ that is a subgraph of $\hat{G}_K$.
\end{corollary}

Furthermore, the following lemma shows that the sum of the lengths of the shortest paths that define $\hat{G}_K$ are linear in the lengths of the paths used to describe the terminals. 

\begin{lemma} \label{lemma:pathLength}
    The number of vertices in $\partial \hat{G}_K$ is $\Oh(N)$.
\end{lemma}

\begin{proof}
    $\partial \hat{G}_K$ comprises shortest paths between terminals and the description of each terminal is a path from the origin to that terminal. By the triangle inequality, the shortest path between neighboring terminals is no longer than the sum of the lengths of their descriptions. Hence, the number of vertices in $\partial \hat{G}_K$ must be linear in $N$.
\end{proof}
Since the time complexity of Algorithm \ref{alg:boundary} is bounded by the number of edges in $\gamma$, Lemma \ref{lemma:pathLength} implies that Algorithm \ref{alg:boundary} runs in time $\Oh(N)$.

\begin{algorithm}
\caption{\textsc{Compute $\partial\hat{G}_K$}}
 \label{alg:boundary}
\begin{algorithmic}
    \Require Sequence of terminals $\langle v_0,\dots,v_{k-1},v_k=v_0\rangle$
    \For{$i\in\{0,k-1\}$}
        \State $\langle s_0, \dots, s_m\rangle \gets S_{v_i,v_{i+1}}$ \Comment{Identify the vertices/edges intersected by $v_iv_{i+1}$}
        \State $\gamma_{1},\gamma_{2} \gets \{\}$ \Comment{Initialize paths from $s_0$ to $s_j$}
        \For{$j\in\{1,m\}$}
            \If {$s_j$ is a vertex}
                \State $\gamma_{1},\gamma_{2} \gets$ shortest extension of $\gamma_1$ or $\gamma_2$ to $s_j$
            \Else 
                \State $s_j$ is an edge $(y_1,y_2)$
                \State $\gamma_{1}\gets $ shortest extension of $\gamma_1$ or $\gamma_2$ to $y_1$
                \State $\gamma_{2}\gets $ shortest extension of $\gamma_1$ or $\gamma_2$ to $y_2$
            \EndIf
        \EndFor
        \State $\gamma^{(i)}\gets \gamma_1$ \Comment{$\gamma^{(i)}$ is a ($v_i,v_{i+1}$)-shortest path}
    \EndFor
    \State $\gamma \gets \bigcup_{i=0}^{k-1}\gamma^{(i)}$
    \Comment{Combine shortest paths to form a boundary walk}
    \State\Return $\gamma$
\end{algorithmic}
\end{algorithm}

\begin{lemma}\label{lemma:tileArea}
    The tiles of $G_{p,q}$ have area $\Theta(p)$.
\end{lemma}
\begin{proof}
    Using the formula for the area of a hyperbolic polygon, the area of a tile is $(p-2)\pi - p(2\pi/q) = 2\pi p (\frac{1}{2} - \frac{1}{p} - \frac{1}{q})$.
    If $p = 3$ or $q = 3$, then the other is at least $7$ and thus $\frac{1}{2} - \frac{1}{p} - \frac{1}{q} \geq \frac{1}{42}$.
    If $p = 4$ or $q = 4$, then the other is at least $5$ and thus $\frac{1}{2} - \frac{1}{p} - \frac{1}{q} \geq \frac{1}{20}$.
    If both $p \geq 5$ and $q \geq 5$, then $\frac{1}{2} - \frac{1}{p} - \frac{1}{q} \geq \frac{1}{10}$.
    Thus, tiles always have area at least $\frac{\pi}{21}p$.
    The upper bound directly follows from the tiles being hyperbolic $p$-gons.
\end{proof}

\begin{lemma}\label{lemma:convexHullSize}
    The convex hull $\conv_G(K)$ has $\Oh(N)$ vertices.
\end{lemma}
\begin{proof}
    Following the same proof as in Lemma~\ref{lemma:pathLength}, the boundary $\partial\conv_G(K)$ has length $\Oh(N)$.
    Thus, $\partial\conv_G(K)$ forms a hyperbolic polygon with $\Oh(N)$ vertices and must have area $\Oh(N)$.
    Lemma~\ref{lemma:tileArea} now implies that $\partial\conv_G(K)$ encloses $\Oh(N / p)$ tiles.
    Each of these is incident to $p$ vertices, giving a bound $\Oh(N)$ on the total number of vertices in $\conv_G(K)$.
\end{proof}

%\skb{If we have time, generalize this algo to the (unique) convex hull.}

\begin{restatable}{lemma}{lemmaHatGKTime}\label{lemma:convexHullTime}
    $\hat{G}_K$ can be computed in $\Oh(N \log N)$ time.
\end{restatable}

\begin{proof}
    For each terminal, we can compute its coordinates in the Beltrami-Klein model in time linear in the number of steps used to describe its location. Hence, we obtain the coordinates of all terminals in $\Oh(N)$ time. Since in the Beltrami-Klein model, hyperbolic lines appear as Euclidean lines, standard convex hull techniques such as Graham's scan \cite{Graham1972AnEA} can then be used to find $\conv_\Hyp(K)$ in $\Oh(n \log n )$ time.
    From there, we can use Algorithm~\ref{alg:boundary} to calculate the sequence of vertices and edges of $\partial\hat{G}_K$ in $\Oh(N)$ time.

    Finally, we need to explicitly construct the graph $\hat{G}_K$ (note that the tiling graph $G_{p,q}$ is infinite and was therefore not stored explicitly).
    For this, we will maintain an associative array $A$ that maps coordinates to the corresponding vertex in $\hat{G}_K$.
    We first add the computed vertices and edges of $\partial\hat{G}_K$ to $\hat{G}_K$.
    Here already, it can be that a vertex appears twice in the sequence, so we recognize and handle this using $A$.
    Next, we fill in the interior with a depth-first search starting from $V(\partial\hat{G}_K)$.
    Since $G_{p,q}$ is not stored explicitly, we cannot mark vertices as one usually would but instead use $A$ to check in $\Oh(\log N)$ time if a vertex has already been explored.
    If the vertex was unexplored, we add it to $\hat{G}_K$ with an edge to the previous vertex, otherwise we only add the edge.
    Depth-first search checks each edge of $\hat{G}_K$ at most twice.
    Since $\hat{G}_K$ is a planar graph with $\Oh(N)$ vertices, it also has $\Oh(N)$ edges and thus this takes $\Oh(N \log N)$ time.
\end{proof}

\section{Bounding the Treewidth of the Convex Hull} \label{sec:bounding}

In this section, we show that both $\hat{G}_K$ and $\conv_G(K)$ have small treewidth.
% and bound their overall number of vertices.
This enables the use of Steiner tree and subset TSP algorithms parameterised by treewidth to solve those problems on $\hat{G}_K$ in polynomial time.
We start with a simple geometric result.

\begin{lemma} \label{lemma:fewInnerFaces} The number of faces of tiling $\{p,q\}$ fully contained in $\conv_\Hyp(K)$ is $\Oh(n/p)$. 
\end{lemma}
\begin{proof}
    The total area of faces fully contained in $\conv_\Hyp(K)$ is upper-bounded by the area of $\conv_\Hyp(K)$, which is upper-bounded by $(n-2)\pi$ as it is a polygon with at most $n$ sides.
    Using Lemma~\ref{lemma:tileArea}, the number of inner faces is now at most $21(n-2)/p$.
\end{proof}

To bound the treewidth of $\conv_G(K)$ (and thereby $\hat{G}_K$), we first combine Lemma~\ref{lemma:fewInnerFaces} with the following result of Kisfaludi-Bak \cite{Kisfaludi-Bak20-Intersection} to get a bound on outerplanarity, which implies a treewidth bound.
Proving this bound both from the primal and the dual tiling yields Theorem~\ref{thm:treewidth}.

\begin{lemma}[Lemma 3.2 of Kisfaludi-Bak \cite{Kisfaludi-Bak20-Intersection}]
\label{lemma:tilingOuterplanar}
    Let $\mathcal T$ be a compact regular tiling of $\Hyp^2$.
    Then the neighborhood graph of any finite tile set $S \subset \mathcal T$ with $|S| \geq 2$ is $(c\log|S|)$-outerplanar, where $c$ is an absolute constant independent of our choice of $\mathcal T$.
\end{lemma}

\thmtw*

\begin{proof}
    Consider the tiles $S$ fully contained in $\conv_\Hyp(K)$; by Lemma~\ref{lemma:fewInnerFaces} $|S| = \Oh(n/p)$.
    Lemma~\ref{lemma:tilingOuterplanar} now implies that the neighborhood graph $N_S$ of $S$ is $\Oh(\log\frac{n}{p})$-outerplanar (unless $|S| \leq 2$; then $N_S$ is $1$-outerplanar).
    Graph $N_S$ is a subgraph of the dual graph $\hat{G}_K^*$ of $\hat{G}_K$.
    In particular, every face of $\hat{G}_K$ that is not in $S$ must have a vertex on the outer face of $\hat{G}_K$.
    Further, by Lemma \ref{lemma:shortestPath} (ii), any face of the convex hull $\conv_G(K)$ that was not in $\hat{G}_K$ must be adjacent to the outer face of $\conv_G(K)$.
    Thus, three iterations of removing all vertices on the outer face of $\conv_G(K)^*$ will give a subgraph of $N_S$.
    Therefore, $\conv_G(K)^*$ has outerplanarity at most $3$ more than $N_S$ and is $\Oh(\log\frac{n}{p})$-outerplanar (or $4$-outerplanar).
    Bodlaender \cite[Theorem 83]{BODLAENDER19981} proves that the treewidth of a $k$-outerplanar graph is at most $3k-1$, so $\conv_G(K)^*$ must have treewidth $\Oh(\log\frac{n}{p})$ (or $11$).
    Finally, Bouchitté, Mazoit and Todinca \cite{BouchitteMT01} prove that the treewidth of a plane graph and its dual differ from each other by at most one, which implies that $\conv_G(K)$ itself has treewidth $\Oh(\log\frac{n}{p})$ (or $12$).

    Now, consider the vertices $V_\text{inner}$ for which the corresponding tile in the dual tiling $G_{q,p}$ lies inside $\conv_\Hyp(K)$.
    By Lemma~\ref{lemma:fewInnerFaces}, $|V_\text{inner}| = \Oh(n / q)$.
    Additionally, induced subgraph $\hat{G}_K[V_\text{inner}]$ is the neighborhood graph of a set of $|V_\text{inner}|$ tiles of the dual tiling, and therefore by Lemma~\ref{lemma:tilingOuterplanar} it is $\Oh(\log\frac{n}{q})$-outerplanar (unless $|V_\text{inner}| \leq 2$; then $\hat{G}_K[V_\text{inner}]$ is $1$-outerplanar).
    Any vertex of $\hat{G}_K$ not in $V_\text{inner}$ is incident to some face of $\hat{G}_K$ not entirely in $\conv_\Hyp(K)$, which means that that face has a vertex on $\partial\hat{G}_K$.
    Thus, after removing the vertices in $\partial\hat{G}_K$, all remaining outer vertices will be on the outer face, so $\hat{G}_K$ has outerplanarity at most $2$ more than $\hat{G}_K[V_\text{inner}]$.
    Next, consider the convex hull $\conv_G(K)$ instead of $\hat{G}_K$.
    Lemma~\ref{lemma:shortestPath}(ii) implies that $\conv_G(K) \setminus \hat{G}_K \subseteq \partial \conv_G(K)$, and thus that $\conv_G(K)$ has outerplanarity at most $1$ more than $\hat{G}_K$.
    Therefore, $\conv_G(K)$ is also $\Oh(\log\frac{n}{q})$-outerplanar (or $4$-outerplanar).
    Here, Bodlaender \cite[Theorem 83]{BODLAENDER19981} implies that $\conv_G(K)$ has treewidth $\Oh(\log\frac{n}{q})$ (or $11$).
    Finally, $\conv_G(K)$ has treewidth at most $\min\{ \Oh(\log\frac{n}{p}),\Oh(\log\frac{n}{q}) \} = \Oh(\log\frac{n}{p + q})$ (or $12$).
%
    % Bodlaender \cite[Theorem 83]{BODLAENDER19981} has shown that the treewidth of a $k$-outerplanar graph is at most $3k-1$ and Kisfaludi-Bak \cite[Lemma 3.2]{Kisfaludi-Bak20-Intersection} has shown that the outerplanarity of a finite tile set $S$ of any tiling of the hyperbolic plane is at most $c\log(|S|)$, where $c$ is independent of $p$ and $q$. Worth noting is that in the case $|S| \leq 1$ we still have outerplanarity (and treewidth) one. The graph of tiles fully contained in $\conv_\Hyp(K)$ has an outerplanarity at most a constant less than $\convGK$. Hence, Lemma \ref{lemma:fewInnerFaces} implies that the treewidth of $\convGK$ is $\max\{1, \Oh(\log \frac{n}{q})\}$.
\end{proof}

Finally, we can combine all our results to get algorithms for Steiner tree and subset TSP.

\algomain*

\begin{proof}
    Using the techniques from Bodlaender, Cygan, Kratsch and Nederlof \cite{BodlaenderCKN15}, given a set of terminals on a graph with $|V|$ vertices and treewidth \textit{tw}, both Steiner tree and subset TSP \cite[Appendix D]{le2019ptassubsettspminorfree} can be solved on it in $2^{\Oh(\textit{tw})}\cdot|V|$ time.
    According to Lemma~\ref{lemma:convGK} it suffices to consider $\hat{G}_K$, which can be computed in $\Oh(N \log N)$ time (Lemma~\ref{lemma:convexHullTime}), has $\Oh(N)$ vertices (Lemma~\ref{lemma:convexHullSize}) and treewidth $\max\{12, \Oh(\log\frac{n}{p + q})\}$ (Theorem~\ref{thm:treewidth}).
    Hence, Steiner tree and subset TSP on $G_{p,q}$ can be solved in $\Oh(N \log N + \poly(\frac{n}{p + q})\cdot N)$ time.
\end{proof}

\bibliographystyle{plainurl} % We choose the "plain" reference style
\bibliography{refs} % Entries are in the refs.bib file

\begin{thebibliography}{10}

\bibitem{trig_formulas}
S.~Axler, F.~W. Gehring, and K.~A. Ribet.
\newblock {\em Foundations of geometry and the non-{E}uclidean plane}.
\newblock Springer New York, 1982.

\bibitem{benedetti1992lectures}
Riccardo Benedetti and Carlo Petronio.
\newblock {\em Lectures on Hyperbolic Geometry}.
\newblock Springer Science \& Business Media, 1992.
\newblock \href {https://doi.org/10.1007/978-3-642-58158-8} {\path{doi:10.1007/978-3-642-58158-8}}.

\bibitem{BlasiusFK16}
Thomas Bl{\"{a}}sius, Tobias Friedrich, and Anton Krohmer.
\newblock Hyperbolic random graphs: Separators and treewidth.
\newblock In {\em 24th Annual European Symposium on Algorithms, {ESA} 2016, August 22-24, 2016, Aarhus, Denmark}, volume~57 of {\em LIPIcs}, pages 15:1--15:16. Schloss Dagstuhl - Leibniz-Zentrum f{\"{u}}r Informatik, 2016.
\newblock \href {https://doi.org/10.4230/LIPICS.ESA.2016.15} {\path{doi:10.4230/LIPICS.ESA.2016.15}}.

\bibitem{structureindependence}
Thomas Bl{\"a}sius, Jean-Pierre von~der Heydt, S{\'a}ndor Kisfaludi-Bak, Marcus Wilhelm, and Geert van Wordragen.
\newblock Structure and independence in hyperbolic uniform disk graphs, 2025.
\newblock To appear in proceedings of SoCG 2025.
\newblock \href {https://arxiv.org/abs/2407.09362} {\path{arXiv:2407.09362}}.

\bibitem{BODLAENDER19981}
Hans~L. Bodlaender.
\newblock A partial \emph{k}-arboretum of graphs with bounded treewidth.
\newblock {\em Theor. Comput. Sci.}, 209(1-2):1--45, 1998.
\newblock \href {https://doi.org/10.1016/S0304-3975(97)00228-4} {\path{doi:10.1016/S0304-3975(97)00228-4}}.

\bibitem{BodlaenderCKN15}
Hans~L. Bodlaender, Marek Cygan, Stefan Kratsch, and Jesper Nederlof.
\newblock Deterministic single exponential time algorithms for connectivity problems parameterized by treewidth.
\newblock {\em Inf. Comput.}, 243:86--111, 2015.
\newblock \href {https://doi.org/10.1016/J.IC.2014.12.008} {\path{doi:10.1016/J.IC.2014.12.008}}.

\bibitem{BouchitteMT01}
Vincent Bouchitt{\'{e}}, Fr{\'{e}}d{\'{e}}ric Mazoit, and Ioan Todinca.
\newblock Treewidth of planar graphs: connections with duality.
\newblock {\em Electron. Notes Discret. Math.}, 10:34--38, 2001.
\newblock \href {https://doi.org/10.1016/S1571-0653(04)00353-1} {\path{doi:10.1016/S1571-0653(04)00353-1}}.

\bibitem{bridson2013metric}
Martin~R. Bridson and Andr{\'e} Haefliger.
\newblock {\em Metric spaces of non-positive curvature}.
\newblock Springer Berlin, Heidelberg, 1999.
\newblock \href {https://doi.org/10.1007/978-3-662-12494-9} {\path{doi:10.1007/978-3-662-12494-9}}.

\bibitem{Cabello25}
Sergio Cabello.
\newblock Testing whether a subgraph is convex or isometric.
\newblock {\em CoRR}, abs/2502.16193, 2025.
\newblock \href {https://arxiv.org/abs/2502.16193} {\path{arXiv:2502.16193}}.

\bibitem{Cannon1984TheCS}
James~W. Cannon.
\newblock The combinatorial structure of cocompact discrete hyperbolic groups.
\newblock {\em Geometriae Dedicata}, 16:123--148, 1984.
\newblock \href {https://doi.org/10.1007/BF00146825} {\path{doi:10.1007/BF00146825}}.

\bibitem{cannonhyperbolic}
James~W. Cannon, William~J. Floyd, Richard Kenyon, and Walter~R. Parry.
\newblock Hyperbolic geometry.
\newblock In James~W. Cannon, editor, {\em Non-{E}uclidean geometry and curvature}, chapter~2. American Mathematical Society, Providence, RI, 2017.
\newblock Two-dimensional spaces. Vol. 3.
\newblock \href {https://doi.org/10.1090/mbk/110} {\path{doi:10.1090/mbk/110}}.

\bibitem{ChepoiDEHV08}
Victor Chepoi, Feodor~F. Dragan, Bertrand Estellon, Michel Habib, and Yann Vax{\`{e}}s.
\newblock Diameters, centers, and approximating trees of delta-hyperbolicgeodesic spaces and graphs.
\newblock In {\em Proceedings of the 24th {ACM} Symposium on Computational Geometry, College Park, MD, USA, June 9-11, 2008}, pages 59--68. {ACM}, 2008.
\newblock \href {https://doi.org/10.1145/1377676.1377687} {\path{doi:10.1145/1377676.1377687}}.

\bibitem{CyganFKLMPPS15}
Marek Cygan, Fedor~V. Fomin, Lukasz Kowalik, Daniel Lokshtanov, D{\'{a}}niel Marx, Marcin Pilipczuk, Michal Pilipczuk, and Saket Saurabh.
\newblock {\em Parameterized Algorithms}.
\newblock Springer, 2015.
\newblock \href {https://doi.org/10.1007/978-3-319-21275-3} {\path{doi:10.1007/978-3-319-21275-3}}.

\bibitem{DUNHAM1986139}
Douglas Dunham.
\newblock Hyperbolic symmetry.
\newblock {\em Computers \& Mathematics with Applications}, 12(1, Part B):139--153, 1986.
\newblock \href {https://doi.org/10.1016/0898-1221(86)90147-1} {\path{doi:10.1016/0898-1221(86)90147-1}}.

\bibitem{wordprocessing}
David B.~A. Epstein, James~W. Cannon, Derek~F. Holt, Silvio~V.F. Levy, Michael~S. Paterson, and William~P. Thurston.
\newblock {\em Word processing in groups}.
\newblock Jones and Bartlett Publishers, 1992.

\bibitem{FominSubexpRectSteiner}
Fedor~V. Fomin, Daniel Lokshtanov, Sudeshna Kolay, Fahad Panolan, and Saket Saurabh.
\newblock Subexponential algorithms for rectilinear steiner tree and arborescence problems.
\newblock 16(2), March 2020.
\newblock \href {https://doi.org/10.1145/3381420} {\path{doi:10.1145/3381420}}.

\bibitem{Graham1972AnEA}
Ronald~L. Graham.
\newblock An efficient algorithm for determining the convex hull of a finite planar set.
\newblock {\em Inf. Process. Lett.}, 1(4):132--133, 1972.
\newblock \href {https://doi.org/10.1016/0020-0190(72)90045-2} {\path{doi:10.1016/0020-0190(72)90045-2}}.

\bibitem{HananGraph}
M.~Hanan.
\newblock On {S}teiner's problem with rectilinear distance.
\newblock {\em SIAM Journal on Applied Mathematics}, 14(2):255--265, 1966.
\newblock URL: \url{http://www.jstor.org/stable/2946265}.

\bibitem{iversen1992hyperbolic}
Birger Iversen.
\newblock {\em Hyperbolic geometry}.
\newblock Number~25 in London Mathematical Society Student Texts. Cambridge University Press, 1992.
\newblock \href {https://doi.org/10.1017/CBO9780511569333} {\path{doi:10.1017/CBO9780511569333}}.

\bibitem{Kisfaludi-Bak20-Intersection}
S{\'{a}}ndor Kisfaludi{-}Bak.
\newblock Hyperbolic intersection graphs and (quasi)-polynomial time.
\newblock In {\em Proceedings of the 2020 {ACM-SIAM} Symposium on Discrete Algorithms, {SODA} 2020}, pages 1621--1638. {SIAM}, 2020.
\newblock \href {https://doi.org/10.1137/1.9781611975994.100} {\path{doi:10.1137/1.9781611975994.100}}.

\bibitem{hyperTSP20}
S{\'{a}}ndor Kisfaludi{-}Bak.
\newblock A quasi-polynomial algorithm for well-spaced hyperbolic {TSP}.
\newblock In Sergio Cabello and Danny~Z. Chen, editors, {\em 36th International Symposium on Computational Geometry, SoCG 2020}, volume 164 of {\em LIPIcs}, pages 55:1--55:15. Schloss Dagstuhl - Leibniz-Zentrum f{\"{u}}r Informatik, 2020.
\newblock \href {https://doi.org/10.4230/LIPICS.SOCG.2020.55} {\path{doi:10.4230/LIPICS.SOCG.2020.55}}.

\bibitem{Kisfaludi-BakML24}
S{\'{a}}ndor Kisfaludi{-}Bak, Jana Masar{\'{\i}}kov{\'{a}}, Erik~Jan van Leeuwen, Bartosz Walczak, and Karol Wegrzycki.
\newblock Separator theorem and algorithms for planar hyperbolic graphs.
\newblock In {\em 40th International Symposium on Computational Geometry, SoCG 2024}, volume 293 of {\em LIPIcs}, pages 67:1--67:17. Schloss Dagstuhl - Leibniz-Zentrum f{\"{u}}r Informatik, 2024.
\newblock \href {https://doi.org/10.4230/LIPICS.SOCG.2024.67} {\path{doi:10.4230/LIPICS.SOCG.2024.67}}.

\bibitem{Kisfaludi-BakW24}
S{\'{a}}ndor Kisfaludi{-}Bak and Geert van Wordragen.
\newblock A quadtree, a steiner spanner, and approximate nearest neighbours in hyperbolic space.
\newblock In Wolfgang Mulzer and Jeff~M. Phillips, editors, {\em 40th International Symposium on Computational Geometry, SoCG 2024}, volume 293 of {\em LIPIcs}, pages 68:1--68:15. Schloss Dagstuhl - Leibniz-Zentrum f{\"{u}}r Informatik, 2024.
\newblock \href {https://doi.org/10.4230/LIPICS.SOCG.2024.68} {\path{doi:10.4230/LIPICS.SOCG.2024.68}}.

\bibitem{KleinM14}
Philip~N. Klein and D{\'{a}}niel Marx.
\newblock A subexponential parameterized algorithm for subset {TSP} on planar graphs.
\newblock In Chandra Chekuri, editor, {\em Proceedings of the Twenty-Fifth Annual {ACM-SIAM} Symposium on Discrete Algorithms, {SODA} 2014}, pages 1812--1830. {SIAM}, 2014.
\newblock \href {https://doi.org/10.1137/1.9781611973402.131} {\path{doi:10.1137/1.9781611973402.131}}.

\bibitem{Kopczynski21}
Eryk Kopczynski.
\newblock Hyperbolic minesweeper is in {P}.
\newblock In {\em 10th International Conference on Fun with Algorithms, {FUN} 2021}, volume 157 of {\em LIPIcs}, pages 18:1--18:7. Schloss Dagstuhl - Leibniz-Zentrum f{\"{u}}r Informatik, 2021.
\newblock \href {https://doi.org/10.4230/LIPICS.FUN.2021.18} {\path{doi:10.4230/LIPICS.FUN.2021.18}}.

\bibitem{dynamicdistances}
Eryk Kopczyński and Dorota Celińska-Kopczyńska.
\newblock Dynamic distances in hyperbolic graphs, 2021.
\newblock \href {https://arxiv.org/abs/2111.01019} {\path{arXiv:2111.01019}}.

\bibitem{le2019ptassubsettspminorfree}
Hung Le.
\newblock A {PTAS} for subset {TSP} in minor-free graphs.
\newblock {\em CoRR}, abs/1804.01588, 2018.
\newblock \href {https://arxiv.org/abs/1804.01588} {\path{arXiv:1804.01588}}.

\bibitem{DBLP:conf/focs/MarxPP18}
D{\'{a}}niel Marx, Marcin Pilipczuk, and Michal Pilipczuk.
\newblock On subexponential parameterized algorithms for steiner tree and directed subset {TSP} on planar graphs.
\newblock In Mikkel Thorup, editor, {\em 59th {IEEE} Annual Symposium on Foundations of Computer Science, {FOCS} 2018}, pages 474--484. {IEEE} Computer Society, 2018.
\newblock \href {https://doi.org/10.1109/FOCS.2018.00052} {\path{doi:10.1109/FOCS.2018.00052}}.

\bibitem{pelayo2013geodesic}
Ignacio~M. Pelayo.
\newblock {\em Geodesic Convexity in Graphs}.
\newblock SpringerBriefs in Mathematics. Springer New York, NY, 1 edition, 2013.
\newblock \href {https://doi.org/10.1007/978-1-4614-8699-2} {\path{doi:10.1007/978-1-4614-8699-2}}.

\bibitem{treewidth}
Neil Robertson and P.D Seymour.
\newblock Graph minors. iii. planar tree-width.
\newblock {\em Journal of Combinatorial Theory, Series B}, 36(1):49--64, 1984.
\newblock \href {https://doi.org/10.1016/0095-8956(84)90013-3} {\path{doi:10.1016/0095-8956(84)90013-3}}.

\bibitem{thurston97three}
William~P. Thurston.
\newblock {\em Three-Dimensional Geometry and Topology, Volume 1}.
\newblock Princeton University Press, 1997.
\newblock \href {https://doi.org/10.1515/9781400865321} {\path{doi:10.1515/9781400865321}}.

\end{thebibliography}

\appendix

\section{Missing Proofs of \Cref{sec:shortestpaths}}\label{sec:appsec3}

\begingroup
\renewcommand\thelemma{\ref{lemma:shortestPath}}
\begin{lemma}
\begin{enumerate}[(ii)]
    \item \lemmaShortestPathPartTwo
\end{enumerate}
\end{lemma}
\endgroup
\begin{proof}
    \noindent \textbf{(ii)} For the first claim, suppose for contradiction such a vertex $w$ exists. Let the sequence of shared vertices between $V(P_1)$ and $V(P_2)$ be $x,v_1,\dots,v_k,y$. Let $v_i,v_{i+1}$ be the two successive shared vertices such that $w$ lies in the the bounded face formed by the subpaths of $P_1$ and $P_2$ between $v_i$ and $v_{i+1}$. These subpaths are also shortest $(v_i,v_{i+1})$-paths. Therefore proving the claim on shortest paths that do not share intermediate vertices suffices to prove the claim on all shortest paths. For the remainder of the proof, we return to referring to the shortest path endpoints as $x,y$ and the paths as $P_1,P_2$. Furthermore, assume that $P_1$ proceeds clockwise around $P_1\cup P_2$ and $P_2$ proceeds counterclockwise. 
    
    Repeat the construction from (i) to obtain the sequence of tiles following $P_1$ and $P_2$, respectively. For the sequence of tiles following $P_1$, start with the first tile clockwise from the first edge of $P_1$ about $x$ and end at the first tile counterclockwise from the last edge of $P_1$ about $y$. For $P_2$, start with the first tile counterclockwise from the first edge of $P_2$ about $x$ and end at the first tile clockwise from the last edge of $P_2$ about $y$. For the path $P_j$, denote the tile centers $c_1^{(j)}, \dots, c_{m^{(j)}}^{(j)}$, the corresponding $(x,y)$-path $Q_j$, the $(x,y)$-curve $\gamma_j$ intersecting the tile centers and the angles $\phi_i^{(j)}$ and $\psi_i^{(j)}$ made by $\gamma_j$ at each tile center $c_i^{(j)}$. Since $\len(P_j)\leq \len(Q_j)$, therefore $\sum_{i=1}^{m^{(j)}}\psi_i^{(j)}\geq\sum_{i=1}^{m^{(j)}}\phi_i^{(j)}$. Additionally, since $\sum_{i=1}^{m^{(j)}}(\psi_i^{(j)} + \phi_i^{(j)})=2\pi m^{(j)}$, we have that $\sum_{i=1}^{m^{(j)}}\psi_i^{(j)}\geq \pi m^{(j)}$.

    The curve $\gamma_1 \cup \gamma_2$ defines a polygon $A$. Since $P_1$ and $P_2$ do not intersect at vertices other than $x$ and $y$, the curves $\gamma_1$ and $\gamma_2$ do not cross, though they may intersect if a tile appears in both sequences. Hence, $\area(A)$ is well-defined as the total area of the bounded face(s) of $A$. Furthermore, for each closed curve $\gamma_j\cup P_j$, the bounded face $F_{\gamma_j\cup P_j}$ does not contain any vertex in $\interior(F_{P_1\cup P_2})$, since the only vertices intersected by $\gamma_j$ are $x$ and $y$. Therefore, all vertices in $\interior(F_{P_1\cup P_2})$ are in $\interior(A)$, so $w\in\interior(A)$ and thus $\area(A)>0$. 
    
    The sum of the internal angles of $A$ is $\geq\sum_{i=1}^{m^{(1)}}\psi_i^{(1)} + \sum_{i=1}^{m^{(2)}}\psi_i^{(2)}$, excluding the internal angles at $u$ and $v$. Using the formula for the area of a hyperbolic polygon:
    \begin{equation*}
        \area(A) \leq \pi(m^{(1)}+m^{(2)}) - \sum_{i=1}^{m^{(1)}}\psi_i^{(1)} - \sum_{i=1}^{m^{(2)}}\psi_i^{(2)} \leq \pi(m^{(1)}+m^{(2)}) - \pi m^{(1)} - \pi m^{(2)} = 0
    \end{equation*}
    which is a contradiction.

    For the second claim, first note that $x$ and $y$ can each be incident to at most $1$ tile contained in $F_{P_1\cup P_2}$. Similarly, each $v\in V(P_1)\cup V(P_2)$ can only be incident to at most two tiles contained in $F_{P_1\cup P_2}$. This leaves $v$ which are only intersected by one of the two paths. Suppose $v$ is incident to $n$ tiles contained in $F_{P_1\cup P_2}$. Let $C$ be the union of the second and $(n-1)$th tiles. Due to the first claim, we are guaranteed that removing the edges incident to $v$, the remaining boundary of $C$ is a subpath of one of the two shortest paths of length $(n-2)(p-2)$. Since the $P_j$s are shortest, the two-edge long path from the start of the above subpath to its end which passes through $v$ cannot be shorter. Hence, we have that $(n-2)(p-2)\leq2$, which gives the second claim.
\end{proof}

\lemmaBottlenecks* 

\begin{proof}
\textbf{(i)}
    First, note that $q \geq 4$ means that the tiles have internal angle $2\pi / q \leq \pi/2$.
    As a consequence, the union $C$ of the tiles surrounding $v$ is a polygon with angles of at most $\pi$ and is thereby convex.
    The two edges of $t_1$ incident to $v$ separate $t_1$ from the rest of $C$, meaning that any line that intersects both $t_1$ and $t_2$ has to intersect one of those two edges.
    If, additionally, $t_1$ and $t_2$ share an edge $e$, then $t_1 \cup t_2$ is already convex and $e$ separates the two tiles.
    Thus, $\ell$ now has to specifically intersect $e$.

\smallskip
\noindent\textbf{(ii)} 
    Let $t_4$ denote the other shared neighbor of $t_1$ and $t_3$.
    We first prove that $\ell$ cannot intersect $t_4$ after $t_3$.
    Let $e$ be the shared edge between $t_1$ and $t_3$, then consider the lines $m_{12}$ and $m_{34}$ separating $t_1$ from $t_2$ and $t_3$ from $t_4$, respectively.
    Lines $m_{12}$ and $m_{34}$ both make the same angle $\frac{2}{3}\pi$ with $e$ and therefore they cannot intersect.
    Thus, after intersecting $m_{12}$ to go from $t_1$ to $t_2$, it is impossible for $\ell$ to intersect $m_{34}$, which is located on the other side of $m_{12}$.
    Consequently, $t_4$ cannot be the neighbor of $t_1$ that $\ell$ intersects after $t_3$.

    Now, consider the hyperbolic convex hull $C$ of the tiles adjacent to $t_1$.
    Each tile other than $t_1$ must have at least one vertex on the boundary of $C$, which means that $t_2 \cup t_1 \cup t_4$ separates $C$ into (at least) two components, one of which contains $t_3$ as the only tile adjacent to $t_1$.
    Since $\ell$ will not intersect any of $t_2 \cup t_1 \cup t_4$ after $t_3$, it will not reach the other component, proving the statement.

\smallskip
\noindent\textbf{(iii)} 
    Again consider the hyperbolic convex hull $C$ of the tiles adjacent to $t_1$.
    Let $t_3$ and $t_4$ be the shared neighbor tiles of $t_1$ and $t_2$.
    By the same reasoning as in (ii), $t_3 \cup t_1 \cup t_4$ separates $C$ into (at least) two components, one of which contains $t_2$ as the only tile adjacent to $t_1$.
    As $\ell$ does not intersect $t_3$ or $t_4$, and $t_1$ is convex, $\ell$ crosses $t_3 \cup t_1 \cup t_4$ exactly once, which means it cannot intersect any tile adjacent to $t_1$ other than $t_2$ after that.
\end{proof}

\lemmapain*

\begin{proof}
    We will additionally require the following formulae for triangle with side lengths $a,b,c$, where the opposite angles have measures $\alpha, \beta,\gamma$. For all triangles, the following four-parts formula holds: \[\cosh c = \cosh a \cosh b - \sinh a \sinh b \cos \gamma.\]For right-angled triangles with hypotenuse of length $c$: 

    \[\tan \alpha = \frac{\tanh a}{\sinh b},\qquad\cosh a = \frac{\cos \beta}{\sin \alpha}, \qquad \cosh c = \cosh a \cosh b = \cot \alpha \cot \beta.\]

    When $q\geq4$, suppose for contradiction there does not exist a shortest $(v,w)$-path that intersects all elements of $S_\ell$ between $s_v$ and $s_w$. By Lemma \ref{lemma:shortestPath} (i), there exists a shortest $(v,w)$-path which uses only edges from $G_\ell$, so this path does not intersect some element $s_i\in S_\ell$ between $s_v$ and $s_w$. This implies that there exists tiles $t_1,t_2$ with interiors intersected by $\ell$ that share a vertex not in $s_i$ and $\ell$ intersects $t_1$, $s_i$ and $t_2$ in that order. If $t_1$ and $t_2$ share an edge, by Lemma \ref{lemma:bottleneck} (i), $s_i$ must be this shared edge, which is a contradiction. Otherwise, $t_1$ and $t_2$ only share one vertex, and this vertex cannot be intersected by $\ell$. Then $s_i$ must be one of the edges intersected by $\ell$ between $t_1$ and $t_2$, but  all edges intersected by $\ell$ between $t_1$ and $t_2$ are incident to this shared vertex (c.f.\ the proof of Lemma \ref{lemma:bottleneck} (i)), which is also a contradiction. Hence $q<4$, so for the remainder of the proof assume that $q=3$.

    By Lemma \ref{lemma:shortestPath} (i), there exists a shortest $(v,w)$-path $P$ that uses only edges from $G_\ell$. Assume for contradiction that $P$ does not intersect all elements of $S_{vw}$. Without loss of generality let $P$ be the shortest $(v,w)$-path that intersects the maximum number of elements of $S_{vw}$.
    \begin{claim*}
        Let $s_i$ be the first element of $S_{vw}$ which $P$ does not intersect.
        Then, $s_i$ is either a vertex incident to three tiles intersected by $\ell$ or an edge with at least one incident vertex incident to three tiles intersected by $\ell$.
    \end{claim*}
    \begin{claimproof}
        Suppose $s_i$ was an edge such that each of its vertices was incident to a tile not intersected by $\ell$. Applying Lemma \ref{lemma:bottleneck} (iii) on both tiles which have $s_i$ as an edge (once in either direction) implies that any $(v,w)$-path must intersect one of the two vertices of $s_i$, which is a contradiction. Hence, we get the claim.
    \end{claimproof}

    Let $t_0,t_1$ refer to the tiles in $G_\ell$ incident to $s_{i-1}$, such that $t_0$ is intersected by $\ell$ before $t_1$ if we travel along $\ell$ in the direction from $s_v$ to $s_w$. If $s_i$ is a vertex, let $t_3$ refer to tile that $\ell$ next intersects the interior of and let $t_2$ refer to the remaining tile incident to $s_i$. If $s_i$ is an edge, let $t_2$ refer to the other tile with $s_i$ as an edge and let $t_3$ refer to the next tile with interior intersected by $\ell$. In both cases, let $u$ be the vertex incident to $t_1,t_2$ and $t_3$.  Without loss of generality, let $t_3$ be the tile immediately clockwise from $t_1$ about $u$.

    \begin{claim*}
        $P$ intersects $z$, the vertex other than $u$ that is incident to both $t_1$ and $t_3$.
    \end{claim*}
    \begin{claimproof}
        In the case where $s_i=u$ is a vertex, this follows directly from the fact that $t_2$ is not in $G_\ell$, so by Lemma \ref{lemma:bottleneck} (iii) any $(v,w)$ path must intersect either $u$ or $z$. In the case where $s_i$ is the edge shared by $t_1$ and $t_2$, by Lemma \ref{lemma:bottleneck} (ii), all $(v,w)$-paths must intersect some vertex in either $t_2$ or $t_3$. However, $P$ intersects $s_{i-1}$, so $P$ intersects a vertex on $t_1$, and any shortest path from $s_{i-1}$ to a vertex in either $t_2$ or $t_3$ must intersect either $s_i$ or $z$.
    \end{claimproof}

    \begin{claim*}
        Consider the first vertex $x$ in $P$ that is incident to $t_1$. One of the following holds:
        \begin{enumerate}
            \item $p$ is odd and $x$ is the vertex on $t_1$ incident to the edge opposite $u$ which is further clockwise about $t_1$.
            \item $p$ is odd and $x$ is the vertex on $t_1$ incident to the edge opposite $u$ which is further counter-clockwise about $t_1$.
            \item $p$ is even and $x$ is the vertex of $t_1$ opposite $u$. 
        \end{enumerate}
    \end{claim*}
    \begin{claimproof}
        This vertex must be incident to $s_{i-1}$. Let the line extension of the edge between $t_2$ and $t_3$ be $\hat{\ell}$. Since $\ell$ intersects the edge between $t_2$ and $t_3$, the point where it intersects $s_{i-1}$ cannot be on the same side of $\hat{\ell}$ as $z$. Therefore, $x$ is $\geq\lfloor p/2\rfloor-1$ edges away from $z$, since since it is either incident to an edge which contains a relative interior point which is on the opposite side of $\hat{\ell}$ as $z$, or is itself on the opposite side of $\hat{\ell}$ as $z$. Furthermore, since $P$ is shortest and intersects the maximum number of elements of $S_{vw}$, the path around $t_1$ from $x$ to $z$ that does not intersect $u$ must be strictly shorter than the path that does intersect $u$. Hence, $x$ is $\leq \lfloor(p-1)/2\rfloor$ edges away from $z$. These two inequalities together imply that we are in one of the three given cases.
    \end{claimproof}

\begin{figure}
    \centering
    \includegraphics{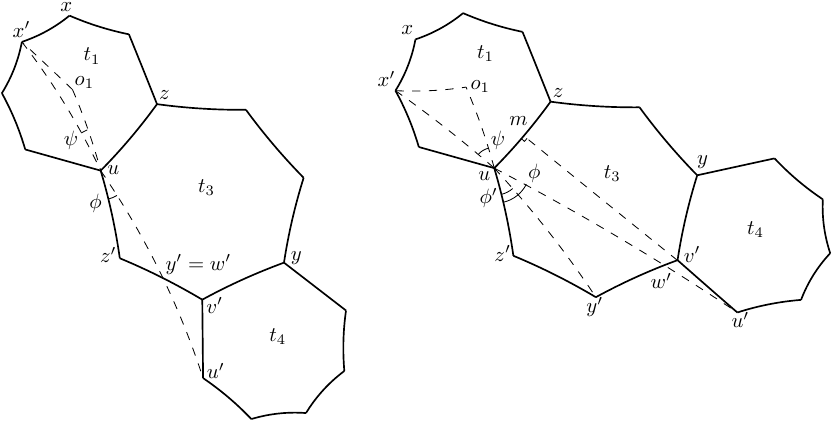}
    \caption{Case 1 (left) and Case 2 (right) when $q=3$ and $p=7$.}
\end{figure}
\begin{figure}
    \centering
    \includegraphics{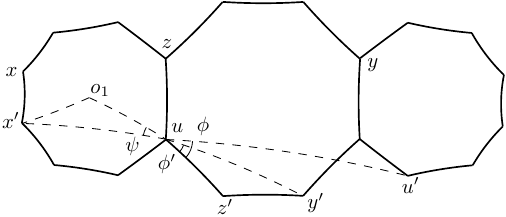}
    \caption{Case 3 when $q=3$ and $p=8$.}
\end{figure}
    
    In all cases, let $y$ be the last vertex of $P$ which is incident to $t_3$. Since $P$ is shortest and intersects the maximum elements of $S_{vw}$, the $zy$-path around $t_3$ is strictly shorter than the $zuy$-path. Furthermore, since $P$ is in $G_\ell$, $y$ must be incident to another tile which is intersected by $\ell$. In each of the cases above, in order for the $zy$-path around $t_3$ to be strictly shorter than the $zuy$-path and the $xzy$-path to be strictly shorter than the $xuy$ path, this next tile must be at most some number of edges away from $z$. We will show that in all cases, the tile which fits the above description cannot be intersected by $\ell$, which gives a contradiction.

    Concretely, in each case, let $t_4$ be the furthest tile clockwise around $t_3$ from $z$ such that $y$ is incident to $t_3$ and $t_4$, the $zy$-path is strictly shorter than the $zuy$-path and the $xzy$-path is strictly shorter than the $xuy$ path. Let the vertex one edge counter-clockwise from $s$ around $t_1$ be $x'$. Let the unique line that intersects $u$ and $x'$ be $\ell'$. It suffices for contradiction to show that $\ell'$ does not intersect $t_4$, since $\ell'$ intersects $\ell$ at some point in $t_1$, so if $\ell'$ does not intersect $t_4$ then neither can $\ell$. For notational convenience, the rest of the proof will refer to $\ell'$ as $\ell$. Without loss of generality, it suffices to show this for the case where $y$ is the furthest distance from $z$ clockwise about $t_3$. When $p$ this occurs when $\len(xzy) = \len(xuy)-2$.  
    
    Furthermore, let $u'$ be the vertex two edges counter-clockwise from $y$ around $t_4$ and let $v'$ be the vertex one edge counter-clockwise from $y$ around $t_4$. Note that $uu'$ intersects $u'v'$ at $u'$ and the line extension of $u'v'$ intersects the center of $t_3$. Hence, $u'v'$ must intersect at least one other tile between $t_3$ and $t_4$. Let $t_5$ be the final tile incident to $v'$. By Lemma \ref{lemma:bottleneck} (ii), since $t_3$ shares a vertex with $t_4$, $uu'$ exactly intersects $t_3,t_5,t_4$ in that order. Hence, $uu'$ intersects the edge shared by $t_3$ and $t_5$. Let this intersection point be $w'$. Since $\angle w'v'u' = 2\pi/3$ and considering $\triangle{w'v'u'}$, we have that $\angle v'u'w'<\pi/3$. Hence, considering the angles on the line extension of $uu'$ about $u'$, the entirety of $t_4$ is contained on or above this line extension. Since $uu'$ intersect $\ell$ at $u$, if $\angle x'uu'<\pi$ then since $\ell$ and the line extension of $uu'$ cannot intersect again, $\ell$ cannot intersect $t_4$. Hence, for a contradiction, it suffices to show that the angle $\angle x'uu'<\pi$ (where $\angle x'uu'$ is measured on the side facing $z$).

\begin{description}
    \item[Case 1.]  The vertex $x$ is the vertex on $t_1$ incident to the edge opposite $u$ and further clockwise about $t_1$. The tile $t_4$ is the tile edge-adjacent to $t_3$ and opposite $z$. Let the vertex incident to $t_3$ and one edge counter-clockwise from $u$ be $z'$. Note that the extension of $z'u$ intersects $o_1$, the center of $t_1$. Let $\psi=\angle x'uo_1$ and $\phi=\angle z'uu'$, such that $\phi>\psi$ implies $\angle x'uu'<\pi$. From the formulas for right-angled triangles, the distance from the center of a tile to its vertex is $r = \cosh^{-1}\left( \cot(\frac{\pi}{3})\cot(\frac{\pi}{p}) \right)$ and half the length of an edge is $a=\cosh^{-1}\left(\frac{\cos(\pi/p)}{\sin(\pi/3)}\right)$. Considering the isosceles triangle $\triangle{x'o_1u}$, we have $\angle x'o_1u = \frac{p-1}{p} \cdot \pi$ so that $\cot(\psi) = \frac{\cosh(r)}{\cot((p-1)\pi/2p)}$. 

    Let $y'$ be the midpoint of the edge incident to $t_3$ and $z'$ but not incident to $u$. Let $\phi'=\angle z'uy'$. Note that $\phi'\leq\phi$, with equality when $p=7$, because the polygon formed by connecting vertices (like $u'$) which are one edge away from $t_3$ is convex and the ray extension of $uy'$ intersects at exactly one of its vertices. By using the four-parts formula on triangle $\triangle{ux'y}$, we have that \[\cot(\phi')=\frac{\sinh(2a)\coth(a)-\cos(2\pi/3)\cosh(2a)}{\sin(2\pi/3)}.\]
    Since $0<\phi',\phi,\psi<\pi/2$ and $\cot$ is strictly decreasing over this interval, it suffices to show that $\cot(\phi)\leq\cot(\phi')<\cot(\psi)$. Using that $\sinh(2a)=2\cosh(a)\sinh(a)$, $\coth(a)=\cosh(a)/\sinh(a)$ and $\cosh(2a)=2\cosh^2(a)-1$, we have that 
    
    \[
        \cot(\phi') =\frac{2\cosh^2(a)-\frac{1}{2}(2\cosh^2(a)-1)}{\frac{1}{2}\sqrt{3}} .
    \]
    Substituting the definition of $a$
    \[
        \cot(\phi')=\frac{2}{\sqrt{3}} \cdot \left( 4\cos^2(\pi/p) - \frac{1}{2} \right) .
    \]
    Finally, using that $\cos^2(\pi/p)\leq1$, we obtain that $\cot(\phi')\leq(7/\sqrt{3})<4.05$. Furthermore, substituting the definition of $r$ into $\cot(\psi)$ gives us $\cot(\psi)=\frac{\cot(\pi/p)}{\sqrt{3}\cot((p-1)\pi/2p)}$. This function is strictly increasing over $p\geq7$, since its derivative \[\frac{\D}{\D p}\cot(\psi)= \pi\cdot \frac{\cot(\pi/p)\csc^2(\pi/2p)+2\cot(\pi/2p)\csc^2(\pi/p)}{2p^2}\] consists only of positive terms over $p\geq7$. When $p=7$, $\cot(\psi)\geq5.25>4.05$, which shows that for all odd $p\geq7$, $\ell$ does not intersect $t_4$.
    
    \item[Case 2.] The vertex $x$ is the vertex on $t_1$ incident to the edge opposite $u$ and further counter-clockwise about $t_1$. The tile $t_4$ is the tile edge-adjacent to $t_3$ and opposite $u$. Define $z',x'$ similarly to above, but define $y'$ to be the vertex incident to $t_3$ which is two edges counter-clockwise from $u$. Now, looking at the isosceles triangle $\triangle{x'o_1u}$ as in Case 1 yields $\cot(\psi)=\frac{\cosh(r)}{\cot((p-3)\pi/2p)}=\frac{\cot(\pi/p)}{\sqrt{3}\cot((p-3)\pi/2p)}$. Considering the isosceles triangle $\triangle{uz'y'}$ and using the four-parts formula, we have that 
    \begin{align*}
        \cot(\phi')&=\frac{\sinh(2a)\coth(2a)-\cos(2\pi/3)\cosh(2a)}{\sin(2\pi/3)} \\
        &=\sqrt{3}\cosh(2a) \\
        &=\frac{8}{\sqrt{3}}\cos^2(\pi/p)-\sqrt{3},
    \end{align*}
    using the same formulas as in the previous case. Similarly, $\phi'<\phi$ for all $p\geq7$ since $uu'$ must intersect the edge of $t_3$ one edge clockwise to the edge shared between $t_3$ and $t_4$. It suffices to show that $\cot(\psi)>\cot(\phi')$. As before, the derivative \[\frac{\D}{\D p}\cot(\psi)=\pi\cdot \frac{2\cot(3\pi/2p)\csc^2(\pi/p)+3\cot(\pi/p)\csc^2(3\pi/2p)}{2p^2},\]contains only positive terms for $p\geq7$. Hence, $\cot(\psi)$ is strictly increasing in $p$. Again bounding $\cos^2(\pi/p)\leq 1$, we have that $\cot(\phi')\leq 5/\sqrt{3}<2.89$. When $p=11$, $\cot(\psi)>4.30$, which proves that $\ell$ does not intersect $t_4$ for $p\geq11$. For $p=9$, the inequality $\cot(\phi')<\cot(\psi)$ can be checked without using the bound on $\cos^2(\pi/p)$ to yield the claim. 
    
    For $p=7$, an exact calculation of $\angle x'uu'$ is required. Let the midpoint of the edge $uz$ be $m$. Note that $\triangle{umu'}$ forms a right-angled triangle, where $\len(um)=a$ and $\len(mu')=2a+b+r$. Hence, $\angle muu'=\tan^{-1}\left(\frac{\tanh(2a+b+r)}{\sinh(a)}\right)$. Therefore, 
    \begin{align*}
    \angle x'uu' &= \psi +\angle o_1um +\angle muu' \\
    &= \cot^{-1}\left(\frac{\cot(\pi/p)}{\sqrt{3}\cot((p-3)\pi/2p}\right) + \pi/3 + \tan^{-1}\left(\frac{\tanh(2a+b+r)}{\sinh(a)}\right) \\
    &< 2.91 <\pi,
    \end{align*}
    which suffices to show a contradiction. 
    
    \item[Case 3.] When $p$ is even, $x$ is the vertex of $t_1$ opposite $u$. The tile $t_4$ is the tile edge-adjacent to $t_3$ that is opposite $t_1$. The calculations of $\cot(\phi')$ are exactly the same as Case 2. In this case, we instead have that  $\cot(\psi)=\frac{\cosh(r)}{\cot((p-2)\pi/2p)}=\frac{\cot(\pi/p)}{\sqrt{3}\cot((p-2)\pi/2p)}$, which again has a first derivative $\frac{\D\cot(\psi)}{\D p} = \frac{2\pi\cot(\pi/p)\csc^2(\pi/p)}{p^2}$ of all positive terms and so is strictly increasing over $p\geq8$. Using $\cos^2(\pi/p)\leq1$, we have that $\cot(\phi')<2.89$ as before and for $p=8$, $\cot(\psi)>3.36$, which suffices to show that for all even $p\geq8$, $\ell$ does not intersect $t_4$.
    \qedhere
\end{description}
\end{proof}

\section{Missing Proofs of \Cref{sec:isomsubgraph}}

\begin{figure}
    \centering
    \includegraphics[width=1.0\linewidth]{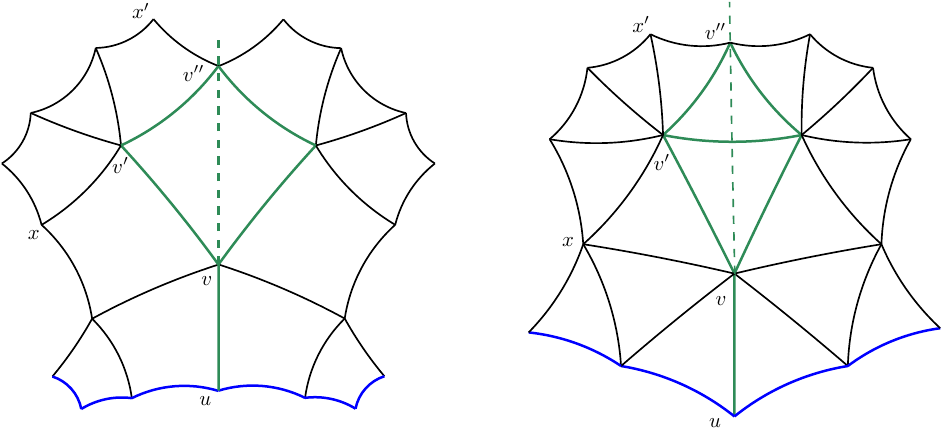}
    \caption{Examples of part of $\partial F$ (blue), line $\ell$ (dashed green) such that $uv\subseteq \ell$ and $G_\ell$ (green), on tiling graphs $G_{4,5}$ (left) and $G_{3,7}$ (right).}
    \label{fig:holefree}
\end{figure}

\lemmaHoleFree*

\begin{proof}
    Suppose there is a bounded face $F$ of $H$ that is not a tile, so there is at least one edge $e=uv$ whose relative interior is inside $\interior(F)$ and $u$ is in $H$. Consider the line $\ell$ of the edge $e$. The line must intersect the boundary of $F$ again either in some vertex $w$ or in some edge $e_w$; in the latter case, let $w$ be some endpoint of $w$. Now, by \Cref{lemma:shortestPath} (i) there is a shortest path $P_0$ from $u$ to $w$ whose first edge is $e=uv$, as depicted in \Cref{fig:holefree}. We claim that every shortest path from $u$ to $w$ uses the edge $uv$.
    
    By \Cref{lemma:shortestPath} (ii) it is sufficient to check detours along tiles that are at least vertex-incident to the shortest path. If $q$ is even, then $\ell$ decomposes into tiling edges (that is, $\ell$ does not intersect the interior of any tile).

    Suppose that $P$ is a shortest path that avoids $v$ by making a detour to the left of $uv$; let $v'$ be the neighbor of $v$ either on $\ell$ (when $q$ is even), or on the tile intersected by $\ell$ to the left of $\ell$. 
    Notice that $P$ must intersect the neighbor $x$ of $v'$ where $v,v',x$ are on the same tile. Thus $P[u,x]$ has length exactly $(p-2)\cdot (\lfloor q/2 \rfloor)-1$. When $p=5$ (and thus $q\geq 4$) or when $p\geq 6$, then this is longer than the path $u,v,v',x$ of length $3$, which is a contradiction. The same holds when $p=4$ and $q\geq 6$. If $p=4$ and $q=5$, then $(p-2)\cdot (\lfloor q/2 \rfloor)-1=3$ and thus stepping back from $x$ to $v'$ is longer than $u,v,v'$; we conclude that the detour must avoid both $v$ and $v'$.
    If $p=3$, then notice that $u,v'$ are not neighbors and not incident to edge-neighboring triangles, thus $u,v,v'$ is the unique shortest path connecting them, so again the detour must also avoid $v$ and $v'$. Let $v''\in \ell$ be the other neighbor of $v'$ on $\ell$. Then either $P$ contains $v'$ ---and we are done by the above argument--- or $P$ also avoids $v''$. Let $x'$ be the shared neighbor of $v'$ and $v''$ on the left of $\ell$ where $P$ must intersect. Then $P[u,x']$ has length exactly $7$ if $p=4$ and $q=5$, and at least $5$ when $p=3$, but this cannot be shortest as $u,v,v',x'$ provides an alternative path of length $3$.

    This concludes the proof that all shortest paths from $u$ to $w$ use the edge $uv$, thus the edge must be contained in $H$.
\end{proof}

\lemmaBoundaryBetweenTerminals*

\begin{proof}
    Suppose for the sake of contradiction that $W_{st}$ is not a shortest $(s,t)$-path. Let $P_{st}$ be a shortest $(s,t)$-path in $\convGK$. Let $\gamma$ be the closed curve formed by $W_{st}$ and $P_{st}$. Partition the vertices contained in $F_\gamma$ into the set of vertices in $V(P_{st})$ and the remaining vertices $A$. Removing $A$ from $\convGK$ leaves us with a geodesic subgraph, since any shortest path between the remaining vertices that previously used a vertex in $A$ intersects $P_{st}$ at least twice and hence can be a shortcut. Therefore, there must have been at least one terminal in $A$. 

    Let the maximal subgraph of $G_K$ with $\gamma$ as its boundary be $G_\gamma$. Without loss of generality, select $P_{st}$ so that $G_\gamma$ contains the fewest vertices. Hence, $P_{st}$ is the unique shortest $(s,t)$-path in $G_\gamma$. For all vertices $v\in A$, since $G_\gamma\subset \convGK$ and $P_{st}$ is a shortest path, there exists a shortest $(s,v)$-path in $G_\gamma$. Let $q\in V(A)$ be a terminal such that applying Lemma \ref{thm:geodesic extension} to $s$ and $q$ identifies a vertex $p\in V(\gamma)$ such that a shortest $(s,p)$-path intersects $q$. Note that $p\notin V(P_{st})$, since the unique shortest path between $s$ and any vertex in $V(P_{st})$ is a subpath of $P_{st}$. Furthermore, since we can select a shortest $(s,p)$-paths for each terminal $q\in V(A)$ such that these paths do not cross, there exists a $q$ such that all terminals in $G_\gamma$ lie (not strictly) on one side of the $(s,p)$-path. Let this path be $P_{sp}$.

    Consider the boundary walk from $s$ to $p$, $W_{sp}$. Note that $|V(W_{sp})|\geq|V(P_{sp})|$, since $P_{sp}$ is a shortest path. Furthermore, $q\notin V(W_{sp})$ since there are no terminals between $s$ and $t$ on $W_{st}$. Let $B$ be the set of vertices in $F_{W_{sp}\cup P_{sp}}$. Then $B\backslash V(P_{sp})$ is non-empty and does not contain any terminals. Removing $B\backslash V(P_{sp})$ from $\convGK$ reduces the number of vertices while maintaining a geodesic subgraph, which contradicts the minimality of $\convGK$.
\end{proof}

\lemmaBoundaryWalk*

\begin{proof}
\textbf{(i)} Suppose for contradiction that $b$ is not a geodesic extension of $W_{uv}$ and there exists a vertex $s$ which is a geodesic extension of $W_{uv}$ and $s$ is further clockwise about $v$ from $W_{uv}$ than $b$. Let $\dist(u,v) = d$, so $\dist(u,s)=d+1$. Since $b$ is not a geodesic extension of $W_{uv}$, $\dist(u,b)\leq d$. Since $\convGK$ is geodesically closed, let $P_{ub}$ be any shortest $(u,b)$-path in $\convGK$. Then $P_{ub}\cup W_{uv} \cup \{vb\}$ forms a closed curve with $s \in \interior(F_{P_{ub}\cup W_{uv} \cup \{vb\}})$. All vertices on this curve have a distance from $u$ that is $\leq d$. However, this contradicts Lemma \ref{thm:geodesic extension}, which identifies a vertex $c$ on this curve such that  $\dist(u,c)\geq d+1$.

\smallskip
\begin{figure}
    \centering
    \includegraphics{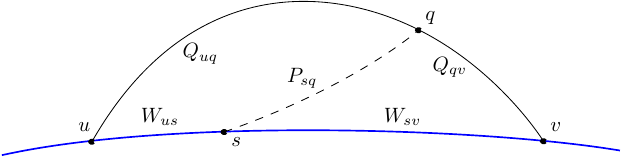}
    \caption{Path $Q_{uv}$ (solid black) connecting two vertices $u,v\in V(\partial\convGK)$. The boundary walk $W_u$ is shown in solid blue.} \label{fig:boundaryShortcut}
\end{figure}
\noindent\textbf{(ii)} We prove the claim by induction on the number of terminals between $u$ and $v$ on $W_{uv}$. If there are no terminals on $W_{uv}$ between $u$ and $v$, then by Lemma \ref{lemma:boundaryBetweenTerminal} $W_{uv}$ is a shortest $(u,v)$-path. This immediately gives the claim.

    If there are terminals on $W_{uv}$ between $u$ and $v$, let $s$ be the first such terminal, as depicted in Figure \ref{fig:boundaryShortcut}. By Lemma \ref{lemma:boundaryBetweenTerminal}, $W_{us}$ is a shortest $(u,s)$-path. Let $\gamma$ be the closed curve $W_{uv}\cup Q_{uv}$.

    \begin{claim} \label{claim:extendCut}
        There exists a vertex $q\in Q_{uv}$ such that the path $W_{us}$ can be extended to a shortest $(u,q)$-path $P_{uq}$ contained in $\gamma$.
    \end{claim}

    Given such a vertex $q$, let $Q_{uq}\cup Q_{qv}=Q_{uv}$ and $W_{us}\cup P_{sq}=P_{uq}$. Since $P_{uq}$ is a shortest path, we have that $\len(Q_{uq}) \geq \len(W_{us}) + \len(P_{sq})$. Since there are strictly fewer terminals on $W_{sv}$ between $s$ and $v$ than on $W_{uv}$ between $u$ and $v$, by the Inductive Hypothesis we have that $\len(P_{sq})+\len(Q_{qv})\geq\len(W_{sv})$. Combining both inequalities gives $\len(Q_{uq}) + \len(P_{sq})+\len(Q_{qv}) \geq \len(W_{us}) + \len(P_{sq}) + \len(W_{sv}) \Rightarrow \len(Q_{uv})\geq \len(W_{uv})$, which proves the lemma.

    To prove Claim \ref{claim:extendCut}, initialize $s'=s$. We maintain the invariants that $s'\in V(W_{sv})$ and $W_{us'}\supseteq W_{us}$ is a shortest $(u,s')$-path, and induct on $\len(W_{s'v})$. If $\len(W_{s'v})=0$ then $q=s'$ suffices. Otherwise, let $c$ be the next vertex after $s'$ on $W_u$. By Lemma \ref{lemma:boundaryWalkLemmas} (i) either $c$ is a geodesic extension of $W_{us'}$ or we can identify a geodesic extension $t$ in $F_{\gamma}$. In the former case, by the Inductive Hypothesis with $s'=c$, we have the claim. In the latter case, if $t\in V(Q_{uv})$ then we are done. Otherwise, by Lemma \ref{thm:geodesic extension} there exists a $a\in V(\gamma)$ such that $W_{us'}$ can be extended to a shortest $(u,a)$-path. If $a\in Q_{uv}$ then we are done. Otherwise, $a\in V(W_{uv})$, so by the Inductive Hypothesis on Lemma \ref{lemma:boundaryWalkLemmas} (ii), the path from $s'$ to $a$ is no longer than $W_{s'a}$. Hence, $W_{ua}$ is a shortest $(u,a)$-path and setting $s'=a$ suffices by the Inductive Hypothesis.
\end{proof}

\section{Missing Proofs of \Cref{sec:computing}}

\lemmaHatGKProperties*

\begin{proof}
    First, we show containment in $\conv_G(K)$. The boundary of $\hat{G}_K$ comprises a union of several shortest paths between terminals. All shortest paths between terminals are included in $\conv_G(K)$. By \Cref{lem:isometricholefree} it follows that $\hat{G}_K$ is a subgraph of $\conv_G(K)$.

    To see that $\hat{G}_K$ is isometrically closed, observe that by construction, any vertices $u,v\in V(\partial\hat{G}_K)$ must be incident to a tile intersected by $\partial \conv_\Hyp(K)$. If $u$ and $v$ are on the same shortest path between terminals that define $\partial\hat{G}_K$ then there immediately exists a shortest ($u,v$)-path in $\hat{G}_K$. Otherwise, let $b_u$ and $b_v$ be any points on $\partial\conv_\Hyp(K)$ that are also on a tile incident to $u$ and $v$, respectively. By Lemma \ref{lemma:shortestPath} (i), there is a shortest $(u,v)$-path in $G_{b_ub_v}$, which we denote $P$. Since $\conv_\Hyp(K)$ is a convex polygon, the line segment $b_ub_v$ is a subset of $\conv_\Hyp(K)$, so all tiles and edges of $G_{b_ub_v}$ intersect $\conv_\Hyp(K)$. The vertices in $G_{b_ub_v}$ incident to tiles and edges in $\interior(\conv_\Hyp(K))$ are guaranteed to be in $\hat{G}_k$. For the tiles in $G_{b_ub_v}$ that intersect $\partial\conv_\Hyp(K)$, consider each subpath $Q_{st}$ of $\partial\hat{G}_K$ corresponding to a shortest path between two terminals $s$ and $t$. If the endpoints of $Q_{st}$ are on the same side of $b_ub_v$, then $Q_{st}$ either does not cross $P$ or crosses $P$ at least twice. For each $Q_{st}$ that crosses $P$ at least twice, the corresponding path between the first and last crossings on $P$ can be replaced to give a new shortest $(u,v)$-path $P'$ that stays within $\hat{G}_K$. Since $b_ub_v$ divides $\conv_\Hyp(K)$ into exactly two parts, the only $Q_{st}$ that can have endpoints on different sides of $b_ub_v$ are the ones that intersect either $u$ or $v$. For these $Q_{st}$, the subpaths of $P$ corresponding to the first and last vertices that intersect each $Q_{st}$ can also be replaced such that $P'$ that stays within $\hat{G}_K$. Hence, $P'$ is a shortest $(u,v)$-path using only edges from $\hat{G}_K$ and thus $\hat{G}_K$ is isometrically closed.
\end{proof}

\end{document}